\documentclass[10pts,twocolumn,romanappendices]{IEEEtran}

\usepackage{cite,calc}

\usepackage{graphicx,colordvi,psfrag,cite}
\usepackage{amsmath,amssymb}
\usepackage{epstopdf}
\usepackage{epsfig}
\usepackage{calc,pstricks, pgf, xcolor}
\makeatletter
\newif\if@restonecol
\makeatother

\usepackage{algorithm2e}
\IEEEoverridecommandlockouts

\newcommand{\abs}[1]{\left\vert#1\right\vert}

\newtheorem{thm}{Theorem}

\newtheorem{lem}{Lemma}

\newtheorem{defn}{Definition}

\begin{document}

\title{Finite-Memory Prediction as Well as \\ the Empirical Mean}
\author{Ronen~Dar  and
        Meir~Feder,~\IEEEmembership{Fellow,~IEEE}
\IEEEcompsocitemizethanks{
The work of Ronen Dar was supported by the Yitzhak and Chaya Weinstein research institute for signal processing.
The work was also partially supported by a grant number 634/09 of the Israeli Science Foundation (ISF).
This paper was presented in part at the IEEE International Symposium on Information Theory, St. Petersburg, Russia, August 2011.

Ronen Dar and Meir Feder are with the Department
of Electrical Engineering-Systems, Tel Aviv University, Ramat Aviv 69978, Israel
(e-mail: ronendar@post.tau.ac.il ; meir@eng.tau.ac.il).
\protect\\

}}

\maketitle

\begin{abstract}
The problem of universally predicting an individual continuous sequence using a deterministic finite-state machine (FSM) is considered.
The empirical mean is used as a reference as it is the constant that fits a given sequence within a minimal square error.
With this reference, a reasonable prediction performance is the regret, namely the excess square-error over the reference loss, the empirical variance.
The paper analyzes the tradeoff between the number of states of the universal FSM and the attainable regret.
It first studies the case of a small number of states. A class of machines, denoted Degenerated Tracking Memory (DTM), is defined and the optimal machine in this class is shown to be the optimal among {\em all} machines for small enough number of states. Unfortunately, DTM machines become suboptimal as the number of available states increases.
Next, the Exponential Decaying Memory (EDM) machine, previously used for predicting binary sequences, is considered.
While this machine has poorer performance for small number of states, it achieves a vanishing regret for large number of states.
Following that, an asymptotic lower bound of $O(k^{-2/3})$ on the achievable regret of any $k$-state machine is derived.
This bound is attained asymptotically by the EDM machine. Furthermore, a new machine, denoted the Enhanced Exponential Decaying Memory machine, is shown to outperform the EDM machine for any number of states.
\end{abstract}
\section{Introduction}
\label{Introduction}
Consider a continuous-valued {\em individual} sequence $x_1,\ldots,x_n$, where each sample is assumed to be bounded in the interval $[a,b]$ but otherwise arbitrary with no underlying statistics.
Suppose that at each time $t$, after observing $x_1^t=x_1,\ldots,x_t$, a predictor guesses the next outcome $\hat{x}_{t+1}$ and incurs a square error prediction loss $(x_{t+1}-\hat{x}_{t+1})^2$.
A reasonable reference for the predictor is the best constant that fits the entire sequence within a minimal square error.
This constant is the empirical mean $\bar{x}=\frac{1}{n}\sum_{t=1}^n x_t$, and its square error is the sequence's empirical variance $\frac{1}{n}\sum_{t=1}^n(x_t-\bar{x})^2$.
Let $\hat{x}_{u,1},\ldots,\hat{x}_{u,n}$ denote the predictions of a (universal) predictor $U$. When the empirical mean is used as a reference, the excess loss of $U$ over the empirical mean, for an individual sequence $x_1^n$, is named the regret:
\begin{equation}
R(U,x_1^n)=\frac{1}{n}\sum_{t=1}^n(x_t-\hat{x}_{u,t})^2-\frac{1}{n}\sum_{t=1}^n(x_t-\bar{x})^2.
\end{equation}

In the setting discussed in this paper, the individual setting, the performance of $U$ is judged by the incurred regret of the worst sequence, i.e., \[\max_{x_1^n} R(U,x_1^n)~.\]
Thus, the optimal $U$ should attain
\[ \min_U \max_{x_1^n} R(U,x_1^n)~. \]
When there are no constraints on the universal predictor, this optimal $U$
is the Cumulative Moving Average (CMA):
\begin{equation}
\label{CMA}
\hat{x}_{t+1}=(1-\frac{1}{t+1})\hat{x}_t+\frac{1}{t+1}x_t,
\end{equation}
where the maximal regret tends to zero with the sequence length $n$ \cite{UniversalPredictionSurvey,UniversalSchemes}. Note that while the reference, the empirical mean predictor, is a constant and needs a single state memory, the CMA predictor is unconstrained and requires an ever growing amount of memory.
A natural question arises - what happens if the universal predictor is constrained to be a finite $k$-state machine?
This is the problem considered in this paper.

Universal estimation and prediction problems where the estimator/predictor is a $k$-state machine have been
explored extensively in the past years. Cover \cite{CoverHypothesisTesting} studied hypothesis testing problem where the tester has a finite memory. Hellman \cite{HellmanFiniteMemoryEstimation} studied the problem of estimating the mean of a Gaussian (or more generally stochastic) sequence using a finite state machine. This problem is closely related to our problem and may be considered as a stochastic version of it: if one assumes that the data is Gaussian, then predicting it with a minimal mean square error essentially boils to estimating its mean. More recently,
the finite-memory universal prediction problem for individual {\em binary} sequences with various loss functions was explored thoroughly in \cite{RajwanFederDCC00,MeronFederDCC04,MeronFederPaper04,IngberFederNonAsy,IngberFederAsy,IngberThesis}. The finite-memory universal portfolio selection problem (that dealt with continuous-valued sequences but considered a very unique loss function) was also explored recently \cite{TavoryFeder}. Yet, the basic problem of finite-memory universal prediction of {\em continuous-valued, individual} sequences with square error loss was left unexplored so far. This paper provides a solution for this problem, presenting such universal predictors attaining a vanishing regret when a large memory is allowed, but also maintaining an optimal tradeoff between the regret and the number of states used by the universal predictor.


The outline of the paper is as follows. In section \ref{sec:ProblemFormulation} we formulate the discussed problem and present guidelines that will be used throughout this paper. Section \ref{chapter:LowNumOfStates} is devoted to universal prediction with a small number of states. We present the class of the Degenerated Tracking Memory (DTM) machines, an algorithm for constructing the optimal DTM machine and a lower bound on the achievable regret. The optimal DTM machine is shown to be the optimal solution among {\em all} machines when a small enough number of states is available. Sections \ref{EDM}, \ref{sec:lowerBoundHigh} and \ref{sec:DesigningEEDM} are devoted to universal prediction using a large number of states. We start in \ref{EDM} by proposing a known universal machine - the Exponential Decaying Memory (EDM) machine - proving asymptotic lower and upper bounds on its worst regret. In section \ref{sec:lowerBoundHigh} we present an asymptotic lower bound on the worst regret of any deterministic $k$-states machine and in section \ref{sec:DesigningEEDM} we present a new machine named the Enhanced Exponential Decaying Memory (E-EDM) machine that can attain any vanishing desired regret while outperforming the EDM machine. In section \ref{sec:Summary} we summarize the results and discuss further research.

\section{Preliminaries}
\label{sec:ProblemFormulation}

We consider universal predictors with continuous-valued input samples that are assumed to be bounded in the interval $[a,b]$. Giving a sequential predictor, we would like to compare the square error incurred by its predictions to the loss incurred by the empirical mean - the best off-line constant predictor. In other words, the reference class comprises all predictors that know the entire sequence in advance, however can predict throughout only a single value. The best predictor among this class is the empirical mean, where its induced loss is the empirical variance.

\begin{defn}
For a given sequence $\{x_1,\ldots,x_n\}$, the excess loss of a universal predictor $U$ with predictions $\{\hat{x}_{1},...,\hat{x}_{n}\}$ over the best constant predictor, the empirical mean  $\bar{x}=\tfrac{1}{n}\sum_{t=1}^n x_t$, is termed the regret of the sequence and is therefore giving by
\begin{equation}
R(U,x_1^n)=\frac{1}{n}\sum_{t=1}^n(x_t-\hat{x}_{t})^2-\frac{1}{n}\sum_{t=1}^n(x_t-\bar{x})^2.
\end{equation}
\end{defn}
We analyze the performance of a universal predictor $U$ by its worst sequence, i.e., by the sequence that induces maximal regret
\begin{equation}
R_{max}(U)=\max_{x_1^n}R(U,x_1^n),
\end{equation}
where we shall take the length of the sequence, $n$, to infinity. The notations $x_1^n$ and $\{{x}_t\}_{t=1}^n$ are used throughout this paper to denote $\{x_1,\ldots,x_n\}$.

The universal predictors considered in this work are memory limited. Finite-State Machine (FSM) is a commonly used model for
sequential machines with a limited amount of storage. We focus here on time-invariant FSM.
\begin{defn}
A deterministic finite-state machine is defined by:
\begin{itemize}
\item An array of $k$ states where $\{S_1,\ldots,S_k\}$ denote the value assigned to each state.
\item The transition of the machine between states is defined by a threshold set $\underline{T}_{~i}=\{T_{i,-m_{d,i}-1},T_{i,-m_{d,i}},\ldots,T_{i,m_{u,i}-1},T_{i,m_{u,i}}\}$ for each state $i$, where $m_{u,i}$ and $m_{d,i}$ are the maximum number of states allowed to be crossed on the way up and down from state $i$, correspondingly. Hence, if at time $t$ the machine is at state $i$ and the input sample $x_t$ satisfies $T_{i,j-1} \leq x_t < T_{i,j}$, the machine jumps $j$ states. Note that the thresholds are non-intersecting, where the union of them covers the interval $[a,b]$ (where each input sample is assumed to be bounded in $[a,b]$).
\item Equivalently, a transition function $\varphi(i,x)$, that is, the next state given that the current state and input sample are $i$ and $x$, can be defined
\begin{equation*}
    \varphi(i,x) = \left\{
    \begin{array}{rl}
    i-m_{d,i} &,T_{i,-m_{d,i}-1} \leq x < T_{i,-m_{d,i}}\\
    i-m_{d,i}+1 &,T_{i,-m_{d,i}}\leq x < T_{i,-m_{d,i}+1}\\
    \vdots\\
    i+m_{u,i}-1 &, T_{i,m_{u,i}-2}\leq x < T_{i,m_{u,i}-1}\\
    i+m_{u,i} &, T_{i,m_{u,i}-1}\leq x < T_{i,m_{u,i}}
    \end{array} \right.
\end{equation*}
\end{itemize}
An FSM predictor works as follows - suppose at time $t$ the machine is at state $i$, then the prediction is $\hat{x}_t=S_i$, the value assigned to state $i$. On receiving the input sample $x_t$, the machine jumps to the next state $\varphi(i,x_t)$. The incurred loss for time $t$ is then $(x_t-\hat{x}_t)^2$.
\end{defn}

Throughout this paper we discuss predictors designed for input samples that are bounded in $[0,1]$. One can easily verify that any FSM that achieves a regret smaller than $R$ for any sequence bounded in $[0,1]$, can be transformed into an FSM that achieves a regret smaller than $R(b-a)^2$ for any sequence bounded in $[a,b]$, by applying the following simple transformation - each state value $S_i$ is transformed into $a+(b-a)S_i$ and each threshold set $\underline{T}_{~i}$ into $a+(b-a)\underline{T}_{~i}$. Thus, all the results presented in this paper can be expanded to the more general case, where each individual sequence is assumed to be bounded in $[a,b]$.

To conclude this section, we provide the definition of a minimal circle and a Theorem that we will use throughout this paper. A version of this Theorem was first given in \cite[Theorem 6.5]{RajwanThesis} - the worst {\em binary} sequence for a given FSM with respect to (w.r.t.) the {\em log-loss} function endlessly rotates the machine in a minimal circle. Here we rederive the proof with emphasis on our case - {\em continuous} sequences and {\em square-loss} function.
\begin{defn}
A circle is a cyclic closed set of
$L$ states/predictions $\{\hat{x}_t\}_{t=1}^L$, if there are input samples $\{x_t\}_{t=1}^L$ that rotate the machine between these states. A minimal circle is a circle that does
not contain the same state more than once. An example is depict in Figure \ref{fig:minimalCircleExmple}.
\end{defn}

\begin{figure}[htb]
\centering
\includegraphics[width=0.5 \columnwidth,height=0.07\textheight]{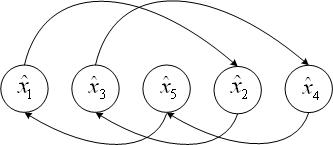}
\caption[Minimal circle - example]{Five states minimal circle - arrows represent the jump at each time $t=1,\ldots,5$. \label{fig:minimalCircleExmple}}
\end{figure}

\begin{thm} \label{thm:problemForm}
The sequence that induces maximal regret over a given FSM, endlessly rotates the machine in a minimal circle.
\end{thm}
\begin{proof}
Let $\{x_t\}_{t=1}^n$ be any sequence of samples and $\{\hat{x}_t\}_{t=1}^n$ the induced sequence of states/predictions on a $k$-states FSM, denoted $U$. Note that $\{\hat{x}_t\}_{t=1}^n$ can be broken into a sequence of minimal circles, denoted $\{c_i\}_{i=1}^m$, and a residual sequence of transient states (which their number is less than $k$). A simple algorithm that generates this sequence of minimal circles works as follows - first search for the first minimal circle in the sequence, that is, the first pair $i$ and $j$ that satisfy $\hat{x}_i=\hat{x}_{j+1}$ where all $\{\hat{x}_t\}_{t=i}^{j}$ are different. Take out these states and their corresponding input samples $\{{x}_t\}_{t=i}^{j}$ to form the first minimal circle $c_1$. Repeat this procedure to construct a sequence of minimal circles. Note that at most $k$ samples are left as a finite residual sequence. Now, denote the length of the minimal circle $c_i$ by $n_i$ and the states and samples that form this circle by $\{\hat{x}_{i,t}\}_{t=1}^{n_i}$ and $\{{x}_{i,t}\}_{t=1}^{n_i}$, respectively. For now assume that there is no residual sequence, then the regret of the complete sequence satisfies
\begin{align}
R(U,x_1^n)&=\frac{1}{n}\sum_{t=1}^n(x_t-\hat{x}_{t})^2-(x_t-\bar{x})^2\\
&=\frac{1}{n}\sum_{i=1}^{m}\sum_{t=1}^{n_i}(x_{i,t}-\hat{x}_{i,t})^2-(x_{i,t}-\bar{x})^2\\
&\leq\frac{1}{n}\sum_{i=1}^{m}\sum_{t=1}^{n_i}(x_{i,t}-\hat{x}_{i,t})^2-(x_{i,t}-\bar{x}_i)^2~,
\end{align}
where $\bar{x}_i=\sum_{t=1}^{n_i} x_{i,t}/{n_i}$ is the empirical mean of minimal circle $c_i$. Let the regret of the minimal circle $c_i$ be $R_i$, then we can write
\begin{align}
R(U,x_1^n)&\leq \frac{1}{n}\sum_{i=1}^{m}n_i R_i~.
\end{align}
Let the minimal circle with the maximal induced regret be $c_j$. Then this regret satisfies $R_j\geq R(U,x_1^n)$. This is true since otherwise, that is, all $R_i$ satisfy $R_i< R(U,x_1^n)$, we get
\begin{align}
R(U,x_1^n)&\leq \frac{1}{n}\sum_{i=1}^{m}n_i R_i\\
&< R(U,x_1^n)~,
\end{align}
which is clearly wrong. Thus, by further noting that for $n\gg k$ the regret induced by the residual sequence is neglectable, and there are finite number of minimal circles in a given FSM, the Theorem can be concluded.
\end{proof}


\section{Designing an optimal FSM with a small number of states}
\label{chapter:LowNumOfStates}

In this section we search for the best universal predictor with relatively small number of states. We start by presenting the optimal machines for a single, two and three states. The optimality is in a sense of achieving the lowest maximal regret using the allowed number of states. We then define in subsection \ref{subsec:theDTMClass} a new class of machines termed the Degenerated Tracking Memory (DTM) machines. This class contains the optimal solutions presented for a single, two and three states. In subsection \ref{sec1} a schematic algorithm for constructing the optimal DTM machine is given. A lower bound on the achievable (maximal) regret of any DTM machine is proven in subsection \ref{subsec:lowerBoundDTM}. We conclude this section in subsection \ref{DTM:conc} by presenting the tradeoff between number of states and regret achieved by the optimal DTM machine. We further discuss the fact that up to a certain number of states, this machine is optimal, not only among the class of DTM machines, but rather among \textbf{all} machines.

\subsection{Single state universal predictor} \label{single}
The problem of finding the optimal single state machine has a trivial solution - from symmetry aspects, the optimal state is assigned with the value $\frac{1}{2}$ and the worst sequence, all $1$'s or $0$'s, incurs a (maximal) regret of $R=\frac{1}{4}$.

\subsection{Two states universal predictor}\label{two}

\begin{figure}[htb]
\centering
\includegraphics[width=0.75 \columnwidth,height=0.075\textheight]{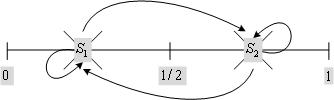}
\caption[Two states machine]{Two states machine described geometrically over the interval $[0,1]$. \label{fig:twoStates_1}}
\end{figure}

A two states machine has two possible minimal circles - zero-step circle (staying at the same state) and two steps circle (toggling between the two states). The lowest maximal regret is achieved when the (maximal) regrets of both minimal circles are equal. Thus, let the lowest state be assigned with the value $S_1=\sqrt{R}$ and a transition threshold $2\sqrt{R}$ and the second state with $S_2=1-\sqrt{R}$ and a transition threshold $1-2\sqrt{R}$. In that case, the regret of the zero-step circles is no more than $R$. Now, let us analyze the regret induced by a sequence $x_1,~x_2,~x_1,~x_2,~...$ that endlessly rotate the machine in the two steps minimal circle. Since the regret is convex in the input samples, maximal regret is attained at the edges of the transition regions, that is, when $x_1=0$ or $x_1=1-2\sqrt{R}$ induces the down-step and $x_2=1$ or $x_2=2\sqrt{R}$ induces the up-step (assuming that the machine starts at the highest state). Therefore there are four combinations that may bring the regret of this minimal circle to maximum. By computing these regrets one gets that the sequence $0,1,0,1,...$ incur the highest regret: $R(U,x_1^n)=R-2\sqrt{R}+3/4$. Equalizing this regret to $R$ results in $R=(\frac{3}{8})^2$ and the maximal regret of both minimal circles is equalized. Therefore the optimal two states machine can be summarized:
\begin{itemize}
\item State values are: \[S_1=\frac{3}{8}~~~,~~~
    S_2=\frac{5}{8}\]
    \item The states transition function satisfies:
    \begin{align*}
    &\varphi(1,x) = \left\{
    \begin{array}{rl}
    1 & ~~\text{if } ~~x <~ \tfrac{3}{4}\\
    2 & ~~\text{otherwise}
    \end{array} \right. \\
    &\varphi(2,x) = \left\{
    \begin{array}{rl}
    1 & ~~\text{if } ~~x <~ \tfrac{1}{4}\\
    2 & ~~\text{otherwise}
    \end{array} \right.
    \end{align*}
\end{itemize}

The worst sequence that endlessly rotates the machine in one of the minimal circles incurs a (maximal) regret of $R=(\frac{3}{8})^2\approx 0.14$. Thus, if the desired regret is smaller than $(\frac{3}{8})^2$ we need to design a machine with more than two states.

\subsection{Three states universal predictor} \label{three}

\begin{figure}[ht]
\centering
\includegraphics[width=0.75 \columnwidth,height=0.075\textheight]{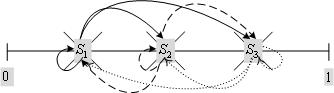}
\caption[Three states machine]{Three states machine described geometrically over the $[0,1]$ axis. \label{fig:threeStates_1}}
\end{figure}

With the same considerations as for the two states machine, the lowest state is assigned with $S_1=\sqrt{R}$ and the upper state with $S_3=1-\sqrt{R}$. From symmetry aspects, the middle state is assigned with $S_2=\frac{1}{2}$. We also note that if a two states jump is allowed from the lower state to the upper state, the sequence $0,1,0,1,...$ toggles the machine between these states. In that case, as was done for the two states machine, the incurred regret is no less than $(\frac{3}{8})^2$. Hence, only a single state jump is allowed, otherwise the three states machine has no gain over the two states machine.
Thus, in the same manner as for the two states machine, one can get that the optimal three states machine satisfies:
\begin{itemize}
\item State values are:
    \[S_1=0.3285~~~,~~~S_2=0.5000~~~,~~~ S_3=0.6715\]
\item The states transition function satisfies:
\begin{align*}
    &\varphi(1,x) = \left\{
    \begin{array}{rl}
    1 & ~~\text{if } ~~x <~ 0.6570\\
    2 & ~~\text{otherwise}
    \end{array} \right. \\
    &\varphi(2,x) = \left\{
    \begin{array}{rl}
    1 & ~~\text{if } ~~x <~ 0.1715\\
    2 & ~~\text{if } ~~0.1715 \leq ~~x <~ 0.8285\\
    3 & ~~\text{otherwise}
    \end{array} \right. \\
    &\varphi(3,x) = \left\{
    \begin{array}{rl}
    2 & ~~\text{if } ~~~x <~ 0.3430\\
    3 & ~~\text{otherwise}
    \end{array} \right.
    \end{align*}
\end{itemize}

The worst sequence that endlessly rotates the machine in one of the minimal circles incurs a (maximal) regret of $R=0.1079$.

Figure \ref{fig:threeStates_2} depict the states and the transition thresholds over the interval $[0,1]$. Note the {{\em hysteresis}} characteristics of the machine, providing ``memory'' or ``inertia'' to the finite-state predictor - an extreme input sample is needed for the machine to jump from the current state, that is, to change the prediction value.

\begin{figure}[htb]
\centering
\includegraphics[width=1\columnwidth]{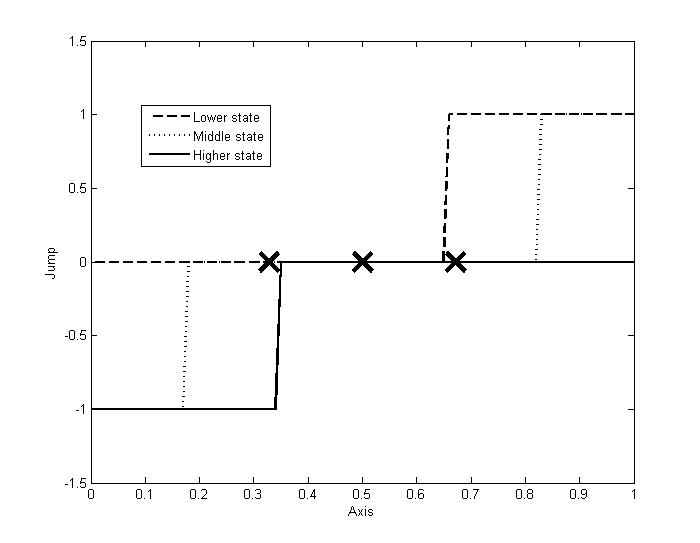}
\caption[Optimal three states machine]{Optimal three states machine described geometrically over the interval $[0,1]$ along with the transition thresholds of the lower state (dashed line), middle state (doted line) and upper state (solid line). The X's represent the value assigned to each state.  \label{fig:threeStates_2}}
\end{figure}

\subsection{The class of DTM machines} \label{subsec:theDTMClass}
We now want to find a more general solution for the best universal predictor with a small number of states. We start by defining a new class of machines and then provide an algorithm to construct the optimal machine among this class. This optimality is in the sense of achieving the lowest maximal regret using the allowed number of states. The optimality of our algorithm among the class of DTM machines is being proved. We Further show that for small enough number of available states, this optimal DTM machine is also optimal among {\textbf {all}} machines.

\begin{defn}
The class of all $k$-states \textbf{\em Degenerated Tracking Memory} (DTM)
machines is of the form:
\begin{itemize}
\item An array of $k$ states - $\{S_{k_l},...,S_{1}\}$ are the states in the lower half  (in descending order where $S_1$ is the nearest state to $\frac{1}{2}$ and $S_i\leq \frac{1}{2}$ for all $1\leq i\leq k_l$), $\{\bar{S}_{1},...,\bar{S}_{k_u}\}$ are the states at the upper half (in ascending order where $\bar{S}_1$ is the nearest state to $\frac{1}{2}$ and $\bar{S}_i> \frac{1}{2}$ for all $1\leq i\leq k_u$), where $k_l+k_u=k$.
\item The maximum down-step in the lower half, i.e., from states $\{S_{k_l},...,S_{1}\}$, is no more than a single state jump. The maximum up-step in the upper half, i.e., from states $\{\bar{S}_1,...,\bar{S}_{k_u}\}$ is no more than a single state jump.
\item A transition between the lower and upper halves is allowed only from and to the nearest states to $\frac{1}{2}$, $S_1$ and $\bar{S}_1$ (implying that the maximum up-jump (down-jump) from $S_1$ ($\bar{S}_1$) is a single state jump).
\end{itemize}
\end{defn}
\label{def1}
An example for a DTM machine is depict in Figure \ref{fig:DtmMachine_example}. Note, however, that the optimal solutions presented before for a single, two and three states, belong to the class of DTM machines.

\begin{figure}[ht]
\centering
\includegraphics[width=1 \columnwidth,height=0.07\textheight]{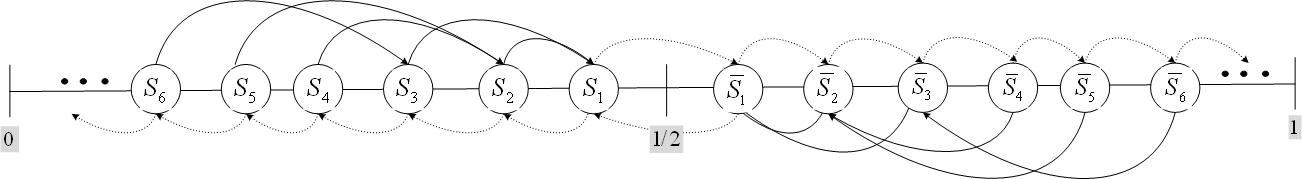}
\caption[Example of a DTM machine]{An example of a DTM machine - note that a transition between the lower and upper halves is allowed only from (and to) $S_1$ and $\bar{S}_{1}$. Arrows represent the maximum up or down jumps from each state. \label{fig:DtmMachine_example}}
\end{figure}

Thus, two constraints define the class of DTM machines - no more than a single state down-step and up-step from all states in the lower and upper halves, respectively, and a transition between these halves is allowed only from and to the nearest states to $\frac{1}{2}$, $S_1$ and $\bar{S}_1$. These constraints facilitate the algorithm for constructing the optimal DTM machine.

\subsection{Constructing the optimal DTM machine}
\label{sec1}

We now present a schematic algorithm
for constructing the optimal DTM machine. Given a desired regret,
$R_d$, the task of finding the optimal DTM machine can be
viewed as a covering problem, that is, assigning the smallest
number of states in the interval $[0,1]$, achieving a regret smaller
than $R_d$ for all sequences. We note that in an optimal $k$-state machine, the upper half of the states is the mirror image of the lower half. The symmetry property arises from the fact that any sequence $\{x_1,...,x_n\}$ can be transformed into the symmetric sequence $\{1-x_1,...,1-x_n\}$. Both sequences induce the same regret if full symmetry between the lower and upper halves is applied. Thus, assuming that the lower half is optimal in sense of achieving the desired regret with the smallest number of states, the upper half must be the reflected image to achieve optimality. Note that this property allows us to design the optimal DTM machine only for the lower half.

The algorithm we present here recursively finds the optimal states' allocation and their transition thresholds. Suppose states $\{S_{i-1},...,S_{1}\}$ in the lower half (in descending order where $S_1$ is the nearest state to $\frac{1}{2}$) and their transition threshold set
$\{\underline{T}_{~i-1},...,\underline{T}_{1}\}$
are given and satisfying regret smaller than $R_d$ for all
minimal circles between them. Our algorithm generates the optimal $S_i$, i.e., the optimal allocation for state $i$, and a threshold set, $\underline{T}_{~i}$, satisfying regret smaller than $R_d$ for all minimal circles starting at that state.

We start by finding $S_1$, the nearest state to $\frac{1}{2}$ in the lower half, in the optimal DTM machine.

\begin{lem} \label{lem4}
In the optimal $k$-states DTM machine for a given desired regret $R_d$, $S_1=\frac{1}{2}$
if $k$ is odd and
\[S_{1}=\max \Big\{1-\sqrt{R_d+\tfrac{1}{4}}~,~2+\sqrt{R_d}-2\sqrt{R_d+\sqrt{R_d}+\tfrac{1}{2}}\Big\}\]
if $k$ is even.
\end{lem}

\begin{proof}
From symmetry aspects $S_1=\frac{1}{2}$ in the optimal DTM machine with odd number of states, otherwise there are more states in one of the halves and the symmetry property presented above does not hold. For even $k$, the nearest state to $\frac{1}{2}$ in the upper half, $\bar{S}_1$, is the mirror image of $S_1$, hence $\bar{S}_{1}=1-S_{1}$. By
definition, only a single state up-jump is allowed from
$S_1$ and only a single state down-jump is allowed from $\bar{S}_1$. Thus, the machine can
be rotated between these states, constructing a two steps minimal circle. Denote by $x_1$ and $x_2$ the samples that induce the up and down jumps, correspondingly. These samples must satisfy the transition thresholds, i.e., \begin{align}
&S_{1}+\sqrt{R_d} \leq ~x_1~\leq 1 \nonumber\\
&0 \leq ~x_2~\leq \bar{S}_{1}-\sqrt{R_d}=1-S_{1}-\sqrt{R_d}~.
\end{align}
Since the regret is a convex function over the input samples, the
regret of a minimal circle is brought to maximum by samples at the edges of the
constraint regions. Thus, in a two steps minimal circle there
are four combinations that may maximize the regret and need to be analyzed. By examining the regrets in all four cases we get that $S_{1}$ must satisfy two constraints $S_1\geq 1-\sqrt{R_d+\tfrac{1}{4}}$ and $S_1 \geq 2+\sqrt{R_d}-2\sqrt{R_d+\sqrt{R_d}+\tfrac{1}{2}}$.\\
\end{proof}

Note that $S_1$ must satisfy $S_{1}\leq\tfrac{1}{2}$ which does not hold for low enough $R_d$, implying a lower bound on the achievable regret of the optimal DTM machine (see section \ref{subsec:lowerBoundDTM}).\\

Now, after presenting the starting state of the algorithm, we present the complete algorithm for constructing the optimal DTM machine:
\begin{enumerate}
\item {\em Set $i=1$ and the corresponded starting state $S_{1}$ (see Lemma \ref{lem4}). Set the maximum up-step from the starting state $m_{u,1}=1$.
\item Set the next state index $i=i+1$.
\item Set the maximal up-step from state $i$ to $m=1$. Find the minimal value that can be assigned to that state with valid threshold set (in sequel we present an algorithm for finding a valid threshold set). Denote this value by $S_{i,m}$ and the threshold set by $\underline{T}_{~i,m}$. Repeat this procedure for all $m=1,\ldots,i-1$ (a jump of $i-1$ states from state $i$ brings the machine to state $S_1$. Remember that an higher jump is not allowed in a DTM machine).
\item Choose the minimal $S_{i,m}$ among all possible maximum up-steps, that is, set
    \[m_{u,i}=\arg\min_{1\leq m \leq i-1}S_{i,m}\]
    \[S_i=S_{i,m_{u,i}}\]
    \[\underline{T}_{~i}=\underline{T}_{~i,m_{u,i}}~.\]
    Thus we have set the parameters of state $i$: assigned value $S_i$, maximum up-jump of $m_{u,i}$ states and transition thresholds $\underline{T}_{~i}$.
\item If $S_i> \sqrt{R_d}$ go to step (2).
\item Set the upper half of the states to be the mirror image of the lower half.}\\
\end{enumerate}

\textbf{Explanations and Comments}:
\begin{itemize}
\item For a given desired regret $R_d$, one should run the algorithm presented above twice - for odd and even number of states with the corresponded starting state, $S_1$. The optimal DTM machine is the one with the least states among the two (differ by a single state).
\item Note that transition thresholds for state $1$ are need to be given - a single state up-jump if the input sample satisfies $x\geq S_1+\sqrt{R_d}$ and a single state down-jump if the input sample satisfies $x\leq S_1-\sqrt{R_d}$. These are the optimal transition thresholds since as the interval for transition is wider the number of possible worst sequences in other minimal circles decreases. Furthermore, with these transition thresholds the maximal regret of a zero-step minimal circle (staying at $S_1$) is $R_d$.
\item A valid threshold set for state $i$ is a set of transition thresholds that satisfy regret smaller than $R_d$ for all minimal circles starting at state $i$.\\
\end{itemize}

To complete the construction of the optimal DTM machine, we still need to present an algorithm for finding the optimal transition thresholds at each iteration (Step $(3)$). Consider states $\{S_{i-1},...,S_{1}\}$ in the lower half and their transition threshold set $\{\underline{T}_{~i-1},...,\underline{T}_{1}\}$
are given and satisfying regret smaller than $R_d$ for all
minimal circles between them. Suppose also $S_i$ and $m$ are given, where $m$ denotes the maximum up-step from state $i$. Note that there are $m+1$ minimal circles starting at state $i$ (depict in Figure \ref{fig:DtmMachine_1}):
\begin{itemize}
\item Zero-step minimal circle (staying at state $i$).
\item For any $2 \leq j \leq m+1$, a minimal circle of $j$ steps - one up-step (of $j-1$ states), $j-1$ down-steps (of a single state).
\end{itemize}
Also note that these $m+1$ minimal circles are within the lower half, that is, within the states $\{S_{i-1},...,S_1\}$ (see Step $(3)$).

\begin{figure}[ht]
\centering
\includegraphics[width=1 \columnwidth,height=0.07\textheight]{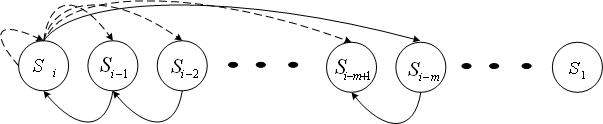}
\caption[DTM machine minimal circle]{$m+1$ possible minimal circles starting at $S_i$, where $m$ is the maximum up-step from state $i$. \label{fig:DtmMachine_1}}
\end{figure}

Let $x_1^j$ be the samples that endlessly rotate the machine in a minimal circle, where $x_1$ induces the up-step from state $i$ and $x_2^j$ induce the down-steps. Since the regret is convex in the input samples, the samples $x_2^j$ that bring the regret to maximum are at the edges of the transition regions, that is, satisfying
\begin{equation} \label{eq:comb}
x_t=\hat{x}_t-\sqrt{R_d} \text{~~or~~} x_t=0 ~~~\forall~~~ 2 \leq t\leq\ j~.
\end{equation}
Thus, there are $2^{j-1}$ combinations of $x_2^j$ that may maximize the regret. Now, given $x_2^j$, Lemma \ref{lem:boundaryUpSample} below provides upper ($C_u(x_2^j)$) and lower ($C_l(x_2^j)$) bounds on $x_1$ so that in this region the induced regret is smaller than $R_d$. Therefore, by computing these bounds for all $2^{j-1}$ combinations of $x_2^j$, one may find a region for $x_1$ in which the regret is lower than $R_d$ for all of these combinations. This region may be given by
\begin{equation}
\tilde{C_l}=\max_{x_2^j\in A_j}  C_l(x_2^j)\leq x_1\leq \min_{x_2^j\in A_j} C_h(x_2^j)=\tilde{C_h}
\end{equation}
where $A_j$ is the set of $2^{j-1}$ combinations of $x_2^j$. Note that this interval is valid only if $\tilde{C_l}\leq \tilde{C_h}$. In that case we can say that the maximal regret of this minimal circle is guaranteed to be lower than $R_d$ and conclude that the transition thresholds for a $j-1$ steps up-jump from state $i$ must satisfy
\begin{align}
&\tilde{C_l}\leq T_{i,j-2}~,\nonumber\\
&T_{i,j-1}\leq \tilde{C_h}~.
\end{align}
Going over all minimal circles, $2\leq j\leq m+1$, results upper and lower bounds, $\tilde{C_l}$ and $\tilde{C_h}$, for each transition threshold. Thus, if a threshold set can be found to satisfy all bounds and to cover the interval $[S_i+\sqrt{R_d}~,~1]$ (that is, $T_{i,m}\geq 1$ and $T_{i,0}\leq S_i+\sqrt{R_d}$), we say that valid transition thresholds for state $i$ were found, otherwise - there are no valid thresholds for the given $S_i$ and $m$.

\begin{lem} \label{lem:boundaryUpSample}
Consider a sequence $x_1^j$ that rotates a DTM machine in a minimal circle starting at state $i$. Given states $\{S_i,\ldots,S_{i-j+1}\}$, the regret is smaller than $R_d$ if $x_1$ satisfies:
\[a(x_2^j)-b(x_2^j) \leq x_1 \leq a(x_2^j)+b(x_2^j)~,\]
where:
\begin{align} \label{eqT1}
&a(x_2^j)=S_i+\sum_{t=2}^j(S_i-x_t)~,   \nonumber\\
&b(x_2^j)=j\sqrt{R_d-\frac{1}{j}\sum_{t=2}^j(S_{i-j+t-1}-S_i)(S_{i-j+t-1}+S_i-2x_t)}~.
\end{align}
\end{lem}
\begin{proof}
Analyzing the regret of the sequence and claiming for regret smaller than $R_d$ results the constrain on $x_1$:
\begin{equation}
\frac{1}{j}\sum_{t=1}^j[(x_t-\hat{x}_t)^2-(x_t-\bar{x})^2]\leq R_d~,
\end{equation}
where $\hat{x}_1=S_{i}$ and $\hat{x}_t=S_{i-j+t-1}$ for $2\leq t\leq j$.
\end{proof}

\bigskip
We can now present the algorithm for finding a threshold set for
state $i$ given $S_i$ and $m$, the maximum up-step:
\begin{enumerate}
\item {\em Find $C_{j,l}$ and $C_{j,h}$ for all $2 \leq j\leq m+1$ as follows:
\begin{align}
&C_{j,l}=\max_{x_2^j\in A_j} ~\bigg\{a(x_2^j)-b(x_2^j)\bigg\}~,\nonumber\\
&C_{j,h}=\min_{x_2^j\in A_j} ~\bigg\{a(x_2^j)+b(x_2^j)\bigg\}~, \label{eq:algoConstrain}
\end{align}
where $a(x_2^j)$ and $b(x_2^j)$ are given in \eqref{eqT1} and $A_j$ is the set of $2^{j-1}$ combinations of $x_2^j$:
\begin{equation}
x_t=S_{i-j+t-1}-\sqrt{R_d} \text{~~or~~} x_t=0 ~~~\forall~~~ 2 \leq t\leq\ j~.
\end{equation}
\item If one of the following constraints does not hold, return and declare that there are no valid thresholds.
    \begin{align}
    &C_{j,l}< C_{j,h} \qquad \forall ~2\leq j\leq m ~,\nonumber\\
    &C_{j+1,l}\leq C_{j,h} \qquad \forall ~2\leq j\leq m ~,\nonumber\\
    &C_{2,l}\leq S_i+\sqrt{R_d}~, \nonumber\\
    &1< C_{m+1,h}~.
    \end{align}
\item Find a valid monotone increasing transition thresholds $\{T_{i,0},\ldots,T_{i,m}\}$ that satisfy:
    \begin{align}
    &C_{j+1,l}\leq T_{i,j-1}\leq C_{j,h} \qquad \forall ~2\leq j\leq m ~,\nonumber\\
    &C_{2,l}\leq T_{i,0}\leq S_i+\sqrt{R_d}~, \nonumber\\
    &1< T_{i,m}\leq C_{m+1,h}~.
    \end{align}
\item Set the transition thresholds for the down-step $\{0,S_i-\sqrt{R_d}\}$.}\\
\end{enumerate}

\textbf{Explanations and Comments}:
\begin{itemize}
\item $C_{j,l}< C_{j,h}$ must be satisfied otherwise there is no $x_1$ that satisfies regret smaller than $R_d$ for all $2^{j-1}$ combinations of $x_2^j$.
\item $C_{j+1,l}\leq C_{j,h}$ must be satisfied otherwise there is no $T_{i,j-1}$ satisfying both $T_{i,j-1}\leq C_{j,h}$ and $C_{j+1,l}\leq T_{i,j-1}$.
\item $T_{i,0}\leq x_1 <T_{i,1}$ induces a single state up-jump, hence, $T_{i,0}$ must satisfy $C_{2,l}\leq T_{i,0}$. Also $T_{i,0}$ must satisfy $T_{i,0}\leq S_i+\sqrt{R_d}$ to ensure regret smaller than $R_d$ for zero-step minimal circle (staying at state $i$).
\item $T_{i,m-1}\leq x_1 <T_{i,m}$ induces $m$ states up-jump, hence, $T_{i,m}$ must satisfy $T_{i,m}\leq C_{m+1,h}$. The transition thresholds must cover the interval $[S_i+\sqrt{R_d},1]$, therefore $T_{i,m}$ must also satisfy $1< T_{i,m}$.
\item This algorithm provides threshold set given the states $\{S_{i-1},...,S_{1}\}$ and $m$, the maximum up-step from state $i$. It also requires $S_i$. Recalling the algorithm for finding $S_i$ - we search for the minimal $S_{i,m}$ with a valid threshold set for a given $m$. Thus, one can provide high $S_{i,m}$ and reduce it until no valid threshold set can be found.\\
\end{itemize}

\begin{thm} \label{thm:optimalDTM}
The algorithm given in this section constructs the optimal DTM machine for a given desired regret, $R_d$, i.e., has the lowest number of states among all DTM machines with maximal regret smaller than $R_d$.
\end{thm}
\begin{proof}
In each iteration the algorithm finds the minimal $S_i$ with a valid threshold set. Note that in DTM machines the transition thresholds for up-steps, $\{T_{i,0},...,T_{i,m_{u,i}}\}$, do not have an impact on regrets of minimal circles other than those starting at state $i$. Thus, given $S_i$, the optimality of these thresholds is only in the sense of satisfying regret smaller than $R_d$ for these minimal circles. As for the down thresholds - an input sample $x$ induces a down-step from state $s$ if satisfies $0\leq x <T_{s,-1}$. As $T_{s,-1}$ is smaller for all states $s=i-1,...,1$ the achievable $S_i$ with a valid threshold set is smaller (the constrains are more relaxed). We choose the smallest $T_{s,-1}$ for all states, i.e., $S_s-\sqrt{R_d}$. Furthermore, each $S_s$ is chosen to be minimal. We further show that optimality is achieved when assigning the minimal value for all states. Consider $\{S_{\lceil \tfrac{k}{2}\rceil},...,S_1\}$ in the lower half are the outputs of the algorithm for a given desired regret $R_d$. Let us examine the case where the assigned value for state $i-1$ is $\tilde{S}_{i-1}$ satisfying $\tilde{S}_{i-1}>S_{i-1}$. We note that the value assigned to state $i-1$ has no impact on the optimality of states $i-2,...,1$. Furthermore, the constrains on the up thresholds of state $i$ depend only on $S_s-S_i$ or $S_s^2-S_i^2$, where $s=i-1,...,1$ (applying $x_t=0$ or $x_t=S_{i-j+t-1}-\sqrt{R_d}$ in Equation \eqref{eqT1}). Since $S_i$ is the minimal value with valid thresholds for $\{S_{i-1},...,S_1\}~$, the minimal value with valid thresholds for $\{\tilde{S}_{i-1},S_{i-2},...,S_1\}$ is not smaller than $S_i$. This holds for all states $\lceil \tfrac{k}{2}\rceil,...,i$ and therefore, choosing $\tilde{S}_{i-1}$ does not reduce the number of states.

Thus, in all aspects optimality is achieved at each iteration in the algorithm by assigning state $i$ with the minimal value $S_i$, down thresholds $\{0,S_i-\sqrt{R_d}\}$ and valid up thresholds.
\end{proof}

\subsection{Lower Bound on the Maximal Regret of DTM Machines}
\label{subsec:lowerBoundDTM}

Here we show that any DTM machine can not attain a maximal regret lower than $(\tfrac{1}{6})^2$. The constraints imposed on this class of machines (as described in section \ref{subsec:theDTMClass}), yield this lower bound.

\begin{thm}
The maximal regret of any DTM machine is lower bounded by
\[R=(\tfrac{1}{6})^2=0.0278~.\]
\end{thm}
\begin{proof}
In an optimal $k$-states DTM machine, where $k$ is even, the starting state $S_1$, must satisfies
\begin{equation} \label{eq:lowerBound}
S_1=\max \{1-\sqrt{R_d+\tfrac{1}{4}},2+\sqrt{R_d}-2\sqrt{R_d+\sqrt{R_d}+\tfrac{1}{2}}\}\leq \tfrac{1}{2},
\end{equation}
implying that if the desired regret satisfies $\sqrt{R_d}<\frac{1}{6}$, then $S_1>\frac{1}{2}$ and
no DTM machine with even number of states can be formed.
We then conclude that also a DTM machine with odd number of
states can not be formed (since otherwise a sub-optimal DTM machine with even number of states
could have been formed by adding another state).
\end{proof}

\subsection{Conclusions} \label{DTM:conc}
In Figure \ref{fig:RvsNumStates1} we present the number of states vs. maximal regret of the machines constructed by the algorithm presented above. Note how the optimal machine can not attain a maximal regret smaller than $1/36$.

\begin{figure}[h]
 \includegraphics[width=1\columnwidth,height=0.25\textheight]{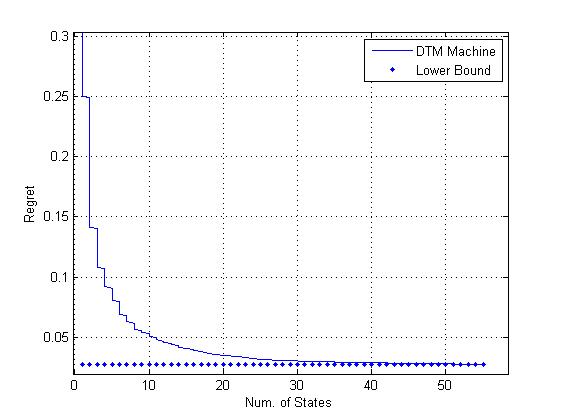}
 \caption[Regret vs. number of states - optimal DTM machine]{Performance of the optimal DTM machine. \label{fig:RvsNumStates1}}
\end{figure}

In this section we started by presenting the optimal solution for machines with a single, two and three states. These solutions belong to the class of DTM machines. Furthermore, one can validate that our algorithm generates for these number of states machines that are identical to these optimal solutions. Thus, in addition to Theorem \ref{thm:optimalDTM}, we can conclude that up to a certain number of states, our algorithm generates the optimal solution among {\textbf {all}} machines. This number, however, is yet unknown.

\section{The Exponential Decaying Memory machine}
\label{EDM}

In the previous section we studied the case of tracking the empirical mean when small number of states are available. In the rest of the paper we shall examine the case of large number of states. We start by proposing the Exponential Decaying Memory (EDM) machine. This machine was presented in \cite{MeronThesis} as a universal predictor for individual {\em binary} sequences. It was further shown that with $k$ states it achieves an asymptotic regret of $O(k^{-2/3})$ compared to the constant predictors class and w.r.t. the
log-loss (code length) and square-error functions. Here we start by describing and adjusting the EDM machine for our case, predicting individual {\em continuous} sequences.
\begin{defn} The \textbf{\emph{Exponential Decaying Memory}} machine
is defined by:
\begin{itemize}
\item $k$ states $\{S_1,...,S_{k}\}$ distributed
uniformly over the interval $[k^{-1/3},1-k^{-1/3}]$.
\item The transition function between states satisfies:
\begin{equation}
\label{EDMeq1}
\hat{x}_{t+1}=Q(\hat{x}_t(1-k^{-2/3})+x_tk^{-2/3})~,
\end{equation}
where $\hat{x}_t$ is the prediction (state) at time $t$ and $Q$ is the quantization function to the nearest state.\\
\end{itemize}
\end{defn}

Note that the spacing gap between states, denoted $\Delta$, satisfies:
\begin{equation} \label{eq:DeltaEDM}
\Delta=\tfrac{1-2k^{-1/3}}{k-1} \sim k^{-1}~,
\end{equation}
and the quantization function satisfies $Q(y)=\hat{x}_{t+1}$,
if $y$ satisfies $\hat{x}_{t+1}-\tfrac{1}{2}\Delta \leq
y < \hat{x}_{t+1}+\tfrac{1}{2}\Delta$. Also note that the EDM machine is a finite-memory approximation of the Cumulative Moving Average predictor given in Equation \eqref{CMA}, where $\frac{1}{t+1}$ is replaced by the constant value $k^{-2/3}$ (which was shown to be optimal in \cite{MeronThesis}).

We now present asymptotic bounds on the regret attained by the EDM machine when used to predict individual \textbf{continuous} sequences.
\begin{thm} \label{thm:EDMregret}
The maximal regret of the $k$-states EDM machine, denoted $U_{EDM_k}$, attained by the worst continuous sequence, is bounded by
\[\tfrac{1}{2}k^{-2/3}+O(k^{-1}) \leq~ \max_{x_1^n}R(U_{EDM_k},x_1^n) \leq \tfrac{17}{4}k^{-2/3}\]
\end{thm}
\begin{proof}
Consider $L$ length sequence $\{x_t\}_{t=1}^L$ that endlessly rotates the machine in a minimal circle of $L$ states $\{\hat{x}_t\}_{t=1}^L$. The input sample at each time $t$ can be written as follows
\begin{equation}
x_t=\hat{x}_t+(P_t\Delta+\delta_t) k^{2/3}~, \label{1}
\end{equation}
where $P_t\in \mathbb{Z}$ denotes the number of states crossed
by the machine at time $t$, $\delta_t$ is a quantization addition that
satisfies $\abs{\delta_t}<\tfrac{1}{2}\Delta$ and has no impact on the jump at time $t$, i.e., has no impact on the prediction at time $t+1$.
Since we examine a minimal circle, the sum of states
crossed on the way up is equal to the sum of states crossed on the
way down, i.e $\sum_{t=1}^L P_t=0$. This means that the empirical mean of the sequence is
\begin{equation}
\bar{x}=\frac{1}{L}\sum_{t=1}^L(\hat x_t+\delta_tk^{2/3})~.
\end{equation}
Now, we can write
\begin{align}
R(U_{EDM_k},x_1^L) &=\frac{1}{L}\sum_{t=1}^L(x_t-\hat{x}_{t})^2-(x_t-\bar{x})^2\\
&=\bar{x}^2+\frac{1}{L}\sum_{t=1}^L(\hat x_t^2-2x_t\hat x_t)~. \label{eq:regSpacQuantFirst}
\end{align}
By Jensen's inequality we have $\bar{x}^2\leq \sum_{t=1}^L (\hat x_t+\delta_tk^{2/3})^2/L$. Applying this and \eqref{1} into Equation \eqref{eq:regSpacQuantFirst} yields
\begin{align}
R(U_{EDM_k},x_1^L)&\leq \tfrac{1}{L}\sum_{t=1}^L \delta_t^2 k^{4/3}-\tfrac{1}{L}\sum_{t=1}^L 2P_t\Delta k^{2/3}\hat{x}_t~. \label{eq:regSpacQuant}
\end{align}
The first term on the right hand side depends only on the quantization of the input samples, $\delta_t$, thus we term it {\em quantization loss}. The second term depends on the spacing gap between states, $\Delta$, thus we term it {\em spacing loss}. Hence, the regret of the sequence is upper bounded by a loss incurred by the quantization of the input samples and a loss incurred by the quantization of the states' values, i.e., the prediction values.
By applying $\abs{\delta_t}<\tfrac{1}{2}\Delta$ we bound the {\em quantization loss}:
\begin{equation}
\text{{\em quantization loss}} = \tfrac{1}{L}\sum_{t=1}^L \delta_t^2 k^{4/3} \leq \tfrac{1}{4}k^{-2/3}~.
\end{equation}
Let us now upper bound the {\em spacing loss}. We define sub-step as a a single state step that is associated with a
full step. For example, a step at time $t$ of $P_t>0$ states consist of $P_t$ sub-steps. We denote these up sub-steps by $\{USS_{t,1},\ldots,USS_{t,P_t}\}$. Note that all of them are associated with a full up-step from state $\hat{x}_t$. Since in a minimal circle the number of states crossed on the way up and down are equal, we can divide all sub-steps into pairs of up and down sub-steps that cross the same state. For example, an up sub-step $USS_{t,j}$ is paired with a down sub-step that crosses the same state. The up sub-step is associated with a full up-step from state $\hat x_t$. The paired down sub-step is associated with a full down-step from a state which we denote by $\hat x_{USS_{t,j}}$. Noting that $P_t$ is positive for up-steps and negative for down-steps, we can write
\begin{align}
-\tfrac{1}{L}\sum_{t=1}^L P_t\hat{x}_t  &= -\tfrac{1}{L}\sum_{t\in \text{\{up steps\}}}P_t\hat{x}_t+\tfrac{1}{L}\sum_{t\in \text{\{down steps\}}}\abs{P_t}\hat{x}_t \nonumber\\
& =\tfrac{1}{L}\sum_{t\in \text{\{up steps\}}} \big( -P_t\hat{x}_t+\sum_{j=1}^{P_t}\hat{x}_{USS_{t,j}}\big)~. \label{eq:subStepEq}
\end{align}
Now, up sub-step $USS_{t,j}$ crosses one of the states between $\hat x_t$ and $\hat x_t+P_t\Delta$. The paired down sub-step has to cross the same state. Since the farthest up or down-step in an EDM machine is $k^{-2/3}$, we can conclude that the paired down sub-step is associated with a full down-step from a state that satisfy $\hat x_{USS_{t,j}}\leq \hat{x}_t+P_t\Delta+k^{-2/3}$. By applying this into Equation \eqref{eq:subStepEq} we get
\begin{align}
-\tfrac{1}{L}\sum_{t=1}^L P_t\hat{x}_t  &\leq \tfrac{1}{L}\sum_{t\in \text{\{up steps\}}}P_t(P_t\Delta+k^{-2/3}) \leq 2\tfrac{k^{-4/3}}{\Delta}~,
\end{align}
where in the last inequality we used $P_t \leq \tfrac{k^{-2/3}}{\Delta}$ (since the farthest step is $k^{-2/3}$). The {\em spacing loss}, thus, satisfies:
\begin{equation}
\text{{\em spacing loss}}=2\Delta k^{2/3}(-\tfrac{1}{L}\sum_{t=1}^L P_t\hat{x}_t) \leq 4k^{-2/3}~.
\end{equation}
By using Theorem \ref{thm:problemForm}, the upper bound is proven.
The proof for the lower bound is given in Appendix \ref{app:lowerBoundTheorem} where we show that there is a sequence that endlessly rotates the $k$-states EDM machine in a minimal circle, incurring a regret of $\frac{1}{2}k^{-2/3}+O(k^{-1})$.

\end{proof}

Note that Theorem \ref{thm:EDMregret} implies that the $k$-state EDM machine achieves a regret smaller than $\tfrac{17}{4}k^{-2/3}$ for any individual continuous sequence bounded in $[0,1]$. Moreover, the regret of the worst sequence, that is, the maximal regret, is at least $\tfrac{1}{2}k^{-2/3}+O(k^{-1})$.

In Figure \ref{fig:RvsNumStates2} the number of states vs. maximal regret achieved by the EDM machine is plotted (regret of $\frac{1}{2}k^{-2/3}$). Also plotted is the performance of the optimal DTM machine. Note that it outperforms the EDM machine for small number of states. Nevertheless, while the achievable (maximal) regret of the optimal DTM machine is lower bounded, the EDM can attain any vanishing regret with large enough number of states.

\begin{figure}[h]
 \includegraphics[width=1\columnwidth,height=0.25\textheight]{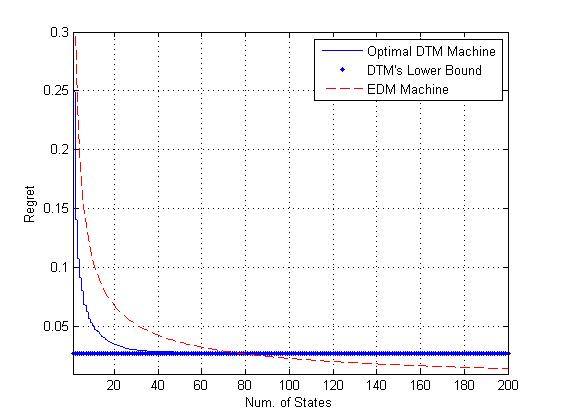}
 \caption[Regret vs. number of states - EDM and DTM machines]{Performance of EDM and optimal DTM machines. \label{fig:RvsNumStates2}}
\end{figure}

\section{Lower bound on the achievable maximal regret of any $k$-states machine}
\label{sec:lowerBoundHigh}

In section \ref{chapter:LowNumOfStates} we have analyzed machines with relatively small number of states. We then examined the case of large number of states and proposed the EDM machine as a universal predictor. We showed that asymptotically, using enough states, it can achieve any vanishing regret. However, is it the optimal solution? Does it attain a desired (maximal) regret with the lowest number of states? In this section we present an asymptotic lower bound on the number of states used by any machine with maximal regret $R$.

\begin{defn}
Given a starting state $S_i$, a \textbf{\em{Threshold Sequence x}}, denoted $TS(x)$, is constructed for any $x$ in the following manner - if the current state is smaller than $x$, next sample in the sequence is $1$ (inducing an up-step), if not, next
sample is $0$ (inducing a down-step).
\end{defn}

For any starting state and any $x$, the constructed $TS(x)$ induces a monotone jumps to the vicinity of $x$ and than rotates the machine in a minimal circle. If the starting state is below $x$, the $TS(x)$ induces monotone up-steps until the machine crosses $x$ (or monotone down-steps if the starting state is above $x$). In the vicinity of $x$ the $TS(x)$ rotates the machine only in a bounded number of states - the lowest possible state is bounded from below by the maximum down-jump from the nearest state to $x$ and the highest possible state is upper bounded by the maximum up-jump from the nearest state to $x$. Therefore, the $TS(x)$ endlessly rotates the machine in a finite number of states, thus inducing a minimal circle. Since the regret induced by the monotone sequence is neglectable, this part can be ignored, and therefore we shall assume that any $TS(x)$ endlessly rotates the machine in a minimal circle, without the monotone part.

\begin{lem} \label{lemMinAway}
Consider an FSM with maximal regret $R$. A $TS(x)$ induces a minimal circle where at least half of its states are within $\tfrac{R}{x}$ from $x$ for any
$~x \leq \tfrac{1}{2}~$ and $~\tfrac{R}{1-x}~$ for any $~x > \tfrac{1}{2}~$.
\end{lem}
\begin{proof}
Let us examine the regret of a $TS(x)$,
where $x \leq \tfrac{1}{2}$, that rotates an FSM, denoted $U$, in a minimal circle of length $L$. Since the empirical mean of the sequence, $\bar{x}$, induces the minimal square error, the regret satisfies
\begin{align}
R(U,x_1^L) & \geq \tfrac{1}{L}\sum_{t=1}^{L}(x_t-\hat{x}_t)^2-(x_t-x)^2 \nonumber\\
& \geq \tfrac{1}{L}\sum_{t=1}^{L}2(x-\hat{x}_t)(x_t-x)~.
\end{align}
We note that by construction, $(x-\hat{x}_t)(x_t-x)$ is positive
for all $t$. Moreover, since $x \leq \tfrac{1}{2}$ and $x_t=1$
for up-steps and $x_t=0$ for down-steps, it follows that:
\begin{equation}
R(U,x_1^L) \geq \tfrac{1}{L}\sum_{t=1}^{L}2\abs{x-\hat{x}_t}x~.
\end{equation}
Hence half of the states have to be within $\tfrac{R}{x}$ from $x$, otherwise we get a regret higher than $R$.
In the same manner it can be shown that for $x > \tfrac{1}{2}$
half of the states have to be within $\tfrac{R}{1-x}$ from $x$.\\
\end{proof}

\begin{lem} \label{lemMinStep}
Consider an FSM with maximal regret $R$. The maximum number of states crossed in an up-step and in a down-step from state $S_i$, for any $i$, must satisfy
\begin{align}
&m_{u,i} \geq \tfrac{1-(S_i+\sqrt{R})}{2\sqrt{R}},\\
&m_{d,i} \geq \tfrac{S_i-\sqrt{R}}{2\sqrt{R}} ~.\label{lemMinStepDown}
\end{align}
\end{lem}
\begin{proof}
See Appendix \ref{app:LemmalemMinStep}.\\
\end{proof}

Note that Lemma \ref{lemMinStep} implies the same lower bound on the achievable regret of any DTM machine, $R\geq (\frac{1}{6})^2$ (as presented in section \ref{chapter:LowNumOfStates}). Any DTM machine allows only a single state down-jump from all states below $\frac{1}{2}$. Thus, a DTM machine may attain maximal regret $R$ if all states below $\frac{1}{2}$ satisfy Equation \eqref{lemMinStepDown} with $m_{d,i}=1$, hence:
\begin{equation} \label{eq:lowerBoundDTM3}
\tfrac{\tfrac{1}{2}-\sqrt{R}}{2\sqrt{R}}\leq 1~.
\end{equation}
Furthermore, Lemma \ref{lemMinStep} provides a lower bound on the maximal regret of any machine that allocates a state $S_i$ with maximum up and down jumps of $m_{u,i}$ and $m_{d,i}$ states.

\bigskip

\begin{thm} \label{thm:nStatesLowerBound}
The number of states in any deterministic FSM with maximal regret $R$, is lower bounded by
\[\tfrac{1}{24}R^{-3/2}+O(R^{-1})~.\]
\end{thm}
\begin{proof}
Consider a $k$-states machine with maximal regret $R$.
Lemma \ref{lemMinAway} implies that for any $x\leq \frac{1}{2}$ there is a $TS(x)$ that forms a minimal circle in the vicinity of $x$ where at least half of the states are within $\tfrac{R}{x}$ from $x$. Since the samples of the $TS(x)$ are either $0$ or $1$, the constructed minimal circle is of at least $m_{u,i}$ states, where $m_{u,i}$ is the maximum up-jump from the nearest state to $x$, denoted state $i$. Thus, there are at least $\frac{1}{2}m_{u,i}$ states within $\tfrac{R}{x}$ from $x$.
Lemma \ref{lemMinStep} implies that the maximum up-step from state $i$
is at least $m_{u,i}=\lceil\tfrac{1-S_i-\sqrt{R}}{2\sqrt{R}}\rceil$ states, where $S_i$ is the assigned value to state $i$.

We define the interval $B(m_u)$ as all $x$'s satisfying
\begin{equation}
m_u=\lceil\tfrac{1-x-\sqrt{R}}{2\sqrt{R}}\rceil~.
\end{equation}
In other words, $B(m_u)$ is the interval \[(1-\sqrt(R)(2m_u+1),1-\sqrt(R)(2m_u-1)]~.\] Note that the length of this interval, $\abs{B(m_u)}$, is always equal to $2\sqrt{R}$. Now, let $N_1$ be the largest integer to satisfy $1-\sqrt(R)(2N_1-1) \geq \tfrac{1}{2}$, and $N_2$ be the smallest integer to satisfy $1-\sqrt(R)(2N_2+1) \leq 0$. We then can write
\begin{equation}
\bigcup_{m_u=N_1}^{N_2}B(m_u)\supseteq[0,\tfrac{1}{2}]~,
\end{equation}
where we note that $\{B(N_1),\ldots,B(N_2)\}$ are non-intersecting intervals. Also note that the smallest value in $B(N_1)$ (that is, $1-\sqrt(R)(2N_1+1)$) is greater than $\tfrac{1}{2}-2\sqrt{R}$. In the same manner, the smallest value in $B(N_1+i)$ (where $i$ is a positive integer), is greater than $\tfrac{1}{2}-2\sqrt{R}(i+1)$.

For $x\in B(m_u)$ there are at least $\tfrac{1}{2}m_u$ states within $\tfrac{R}{x}$ from $x$. Therefore, in the interval $B(m_u)$ there are at least \[\min_{x\in B(m_u)}\tfrac{\abs{B(m_u)}}{R/x}\tfrac{1}{2}m_u\] states. Using the fact that in an optimal machine the minimal number of states in the lower and upper halves is equal (see Section \ref{sec1}), we can conclude that $k$, the number of states, satisfies
\begin{align}
k & \geq 2\sum_{m_u =N_1+1}^{N_2-1}\min_{x\in B(m_u)}\tfrac{\abs{B(m_u)}}{R/x}\tfrac{1}{2}m_u \nonumber\\
&= \sum_{m_u =N_1+1}^{N_2-1}\min_{x\in B(m_u)}\tfrac{2\sqrt{R}}{R/x}\lceil\tfrac{1-x-\sqrt{R}}{2\sqrt{R}}\rceil \nonumber\\
&\geq R^{-1}\sum_{m_u =N_1+1}^{N_2-1}\min_{x\in B(m_u)}x(1-x-\sqrt{R})~. \label{eq:last}
\end{align}
The function $x(1-x-\sqrt{R})$ is concave and has a single maximum point at $\tfrac{1}{2}(1-\sqrt{R})$. Thus, $\min_{x\in B(m_u)}x(1-x-\sqrt{R})$ is attained at the smallest value in the interval $B(m_u)$ (that is, $1-\sqrt(R)(2m_u+1)$). As was mentioned before this value is greater than $\tfrac{1}{2}-2\sqrt{R}(m_u-N_1+1)$ and therefore this further minimizes the function $x(1-x-\sqrt{R})$. Thus, we can write
\begin{align}
k & \geq \tfrac{1}{2}R^{-3/2}\sum_{i=2}^{\lfloor 1/(4\sqrt{R})\rfloor}2\sqrt{R} (\tfrac{1}{2}-2\sqrt{R}i) (\tfrac{1}{2}+2\sqrt{R}i-\sqrt{R}) \nonumber\\
& \geq \tfrac{1}{24}R^{-3/2}-\tfrac{7}{16}R^{-1}+\tfrac{7}{12}R^{-1/2}+2~.
\end{align}
This concludes the proof.
\end{proof}

Note that Theorem \ref{thm:nStatesLowerBound} implies that a $k$-states FSM can not attain maximal regret smaller than
\begin{equation} \label{lowerBound1}
(24k)^{-2/3}+O(k^{-1})~.
\end{equation}

\section{Enhanced Exponential Decaying Memory machine}
\label{sec:DesigningEEDM}

In Section \ref{EDM} we showed that the EDM machine can achieve any maximal regret, as small as desired. In this section we present a new FSM named
the \textbf{\emph{Enhanced Exponential Decaying Memory}} (E-EDM) machine. We prove that it outperforms the EDM machine and better approaches the lower bound presented in the previous section.

\subsection{Designing the E-EDM machine}
\label{sec:DesigningEEDMAlgo}
The algorithm for constructing the E-EDM machine for a desired regret $R_d$, is as follows.

\begin{itemize} \label{algEEDM}
\item Set $R=\tfrac{R_d}{2}$.
\item Divide the interval $[0,1]$ into segments, denoted $A(m_u,m_d)$, where each contains all $x$'s satisfying both
\begin{align}
&m_u=\lceil \tfrac{1-x-\sqrt{R}}{2\sqrt{R}}\rceil~, \nonumber\\
&m_d=\lceil \tfrac{x-\sqrt{R}}{2\sqrt{R}}\rceil~. \label{maxUpDown}
\end{align}
Note that these segments are non-intersecting.
\item Linearly spread states in each segment $A(m_u,m_d)$ with a $\Delta(m_u,m_d)$ spacing gap between them, where
\begin{equation}
\Delta(m_u,m_d)=\tfrac{\sqrt{R}}{2m_{u}\cdot m_{d}}~.
\end{equation}
\item Assign all states in segment $A(m_u,m_d)$ with maximum up and down jumps of $m_u,~m_d$ states, correspondingly. Note that according to Lemma \ref{lemMinStep}, these are the minimal maximum jumps allowed in order to achieve maximal regret smaller than $R$.
\item Assign transition thresholds for each state $i$ as follows:
\begin{equation}
T_{i,j}=S_i+(2j+1)\sqrt{R} \quad \forall \quad -m_{d,i}\leq j \leq m_{u,i}~,
\end{equation}
that is, if the machine at time $t$ is at state $i$, it jumps $j$ states if the current outcome, $x_t$, satisfies:
\begin{equation}
S_i+(2j-1)\sqrt{R} \leq x_t < S_i+(2j+1)\sqrt{R}~.
\end{equation}
Note that as required, the transition thresholds cover the $[0,1]$ axis (arises from the chosen maximum up and down jumps).
\item We further need to guarantee the desired regret when the machine traverses between segments. Consider two adjacent segments $A(m_{u,1},m_{d,1})$ and $A(m_{u,2},m_{d,2})$ and suppose the spacing gap in the second segment is smaller. Add states to the first segment such that the
closest $m_{u,1}+m_{d,1}$ states to the second segment have a spacing gap of
$\Delta(m_{u,2},m_{d,2})$.
It can be shown that at most two states need to be added to each segment. Figure \ref{fig:EEDM1} depict the spacing gap in two adjacent segments.
\end{itemize}

\begin{figure}[ht]
\centering
\includegraphics[width=1 \columnwidth]{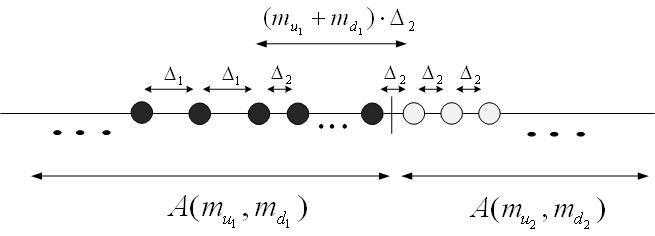}
\caption[E-EDM machine - spacing Gap]{Spacing gap of the E-EDM machine. Adjacent segments $A(m_{u,1},m_{d,1})$ and $A(m_{u,2},m_{d,2})$ with spacing gap $\Delta_s=\tfrac{\sqrt{R}}{2m_{u,s}m_{d,s}}$ where $s=1,2$ and $\Delta_2<\Delta_1$. Note that the spacing gap between the highest $m_{u,1}+m_{d,1}$ states in segment $A(m_{u,1},m_{d,1})$ is $\Delta_2$ while the maximum up and down jumps from these states are $m_{u,1}$ and $m_{d,1}$ states. \label{fig:EEDM1}}
\end{figure}

Recall that the transition thresholds in the EDM machine are $T_{i,j}=S_i+(j+\tfrac{1}{2})\Delta k^{2/3}$. Since $\Delta\sim k^{-1}$, if we take the desired regret to be $R_d=\frac{1}{2}k^{-2/3}$, that is, $R=\frac{1}{4}k^{-2/3}$, we get that the transition thresholds in the E-EDM machine are identical to those defined for the EDM machine. Furthermore, recall that according to Theorem \ref{thm:EDMregret}, the maximal regret of the $k$-states EDM machine is greater than $\frac{1}{2}k^{-2/3}$. Thus, the new machine presented here achieves lower maximal regret by better allocating the states - the states of the EDM are uniformly distributed over the interval $[0,1]$ while in the E-EDM machine the interval $[0,1]$ is divided into segments and states are uniformly distributed with a different spacing in each segment. This will be proved more rigorously in sequel.

We shall now prove that the maximal regret in an E-EDM machine, constructed by the algorithm above, indeed is no more than the desired regret $R_d$.

\begin{thm}
\label{thm:E-EDMRegret}
The construction of the E-EDM machine according to the algorithm \ref{algEEDM}, yields a machine with maximal regret that is no more than $R_d$.
\end{thm}
\begin{proof}
Consider a sequence $x_1^L$ that endlessly rotates the E-EDM machine (denoted $U_{E-EDM}$) in a minimal circle of $L$ states $\hat{x}_1^L$. Each input sample $x_t$ can be written as follows:
\begin{equation}
x_t=\hat{x}_t+2\sqrt{R}\cdot P_t+\delta_t~,
\end{equation}
where $P_t$ is the number of states the machine crosses at time $t$ ($-m_{d}\leq P_t \leq m_{u}$) and $\delta_t$ satisfies $\delta_t\leq \sqrt{R}$ and can be regarded as a quantization addition that has no impact on the jump at time $t$, i.e., has no impact on the next prediction. Since we examine a minimal circle, the sum of states
crossed on the way up is equal to the sum of states crossed on the
way down, i.e $\sum_{t=1}^L P_t=0$. By applying this and Jensen's
inequality, the regret of the sequence satisfies:
\begin{align}
R(U_{E-EDM},x_1^L) & \leq \tfrac{1}{L}\sum_{t=1}^L\delta_t^2 -4\sqrt{R}\tfrac{1}{L}\sum_{t=1}^L P_t(\hat{x}_t-\hat{x}_1)~. \label{lemRegret:eq2}
\end{align}
We term the first loss in the right hand side of Equation \eqref{lemRegret:eq2} {\em quantization loss} (since it depends only on $\delta_t$, the quantization of the input sample, $x_t$). By applying $\delta_t\leq \sqrt{R}$ we get:
\begin{equation}
\text{\em quantization loss}=\tfrac{1}{L}\sum_{t=1}^L\delta_t^2 \leq R~.
\end{equation}
We term the second loss in the right hand side of Equation \eqref{lemRegret:eq2} {\em spacing loss} (since $\hat{x}_t-\hat{x}_1$ depends only on the spacing gap between states).
Thus, as we sowed for the EDM machine, the regret of the sequence is upper bounded by a loss incurred by the quantization of the input samples and a loss incurred by the quantization of the states' values, i.e., the prediction values.
\begin{lem} \label{lem:spacingLossWithinSegment}
For any sequence $x_1^L$ that endlessly rotates the E-EDM machine in a minimal circle of states $\hat{x}_1^L$, where the spacing gap between all states is identical, the spacing loss is smaller than $R$ satisfying:
\begin{equation}
\text{\em spacing loss}=-4\sqrt{R}\tfrac{1}{L}\sum_{t=1}^L P_t(\hat{x}_t-\hat{x}_1) \leq R~.
\end{equation}
\end{lem}
\begin{proof}
See Appendix \ref{app:LemmaspacingLossWithinSegment}.
\end{proof}

\begin{lem} \label{lem:spacinglossSegments}
For any sequence $x_1^L$ that rotates the E-EDM machine in a minimal circle of states $\hat{x}_1^L$, where the spacing gap is not equal between all states, the spacing loss is smaller than $R$ satisfying:
\[\text{\em spacing loss}=-4\sqrt{R}\tfrac{1}{L}\sum_{t=1}^L P_t(\hat{x}_t-\hat{x}_1) \leq R~.\]
\end{lem}
\begin{proof}
See Appendix \ref{app:LemmaspacinglossSegments}.
\end{proof}
Since $R=\tfrac{R_d}{2}$ and by applying Theorem \ref{thm:problemForm} the proof is concluded.\\
\end{proof}

\subsection{Performance Evaluation}
The following Theorem gives the number of states used by an E-EDM machine designed with a desired regret $R_d$.

\begin{thm} \label{thm:nStatesEEDM}
The number of states in an E-EDM machine designed to achieve maximal regret smaller than $R_d$ is
\[ \tfrac{1}{12}(\tfrac{R_d}{2})^{-3/2}+O(R_d^{-1})~. \]
\end{thm}
\begin{proof}
See Appendix \ref{app:TheoremnStatesEEDM}.
\end{proof}

Theorem \ref{thm:EDMregret} implies that the asymptotic worst regret of the $k$-states EDM machine is at least $\tfrac{1}{2}k^{-2/3}$.
Thus, the number of states in an EDM machine with maximal regret $R_d$, is at least $(2R_d)^{-3/2}$ states.
Theorem \ref{thm:nStatesLowerBound} implies that the asymptotic number of states of any deterministic
FSM with maximal regret $R_d$ is at least $\tfrac{1}{24}R_d^{-3/2}$.
Theorem \ref{thm:nStatesEEDM} implies that the asymptotic number of states in an E-EDM machine with maximal regret $R_d$ is $\tfrac{1}{12}(\tfrac{R_d}{2})^{-3/2}$.
Thus we can conclude that:
\begin{enumerate}
\item For a given desired regret, the E-EDM machine
    outperforms the EDM machine in number of states
    by a factor of:
\[\tfrac{\tfrac{2^{3/2}}{12}R_d^{-3/2}}{(2R_d)^{-3/2}}=\tfrac{2}{3}~,\]
i.e., uses only $\frac{2}{3}$ of the states needed for the EDM machine to achieve the same maximal regret.
\item For a given desired regret, the E-EDM machine approaches the lower bound
    with a factor of about:
\[\tfrac{\tfrac{2^{3/2}}{12}R_d^{-3/2}}{\tfrac{1}{24}R_d^{-3/2}}=2^{5/2}=5.6~.\]
\end{enumerate}
In Figure \ref{fig:RvsNumStatesAsyPerf} we plot the (maximal) regret attained by the EDM and E-EDM machines as a function of the number of states, together with the lower bound given in Theorem \ref{thm:nStatesLowerBound}. Note that for a large number of states the E-EDM machine indeed outperforms the EDM machine by a factor of $\sim\frac{2}{3}$ and approaches the lower bound with a factor of $\sim6$.

\begin{figure}[t]
 \includegraphics[width=1\columnwidth,height=0.25\textheight]{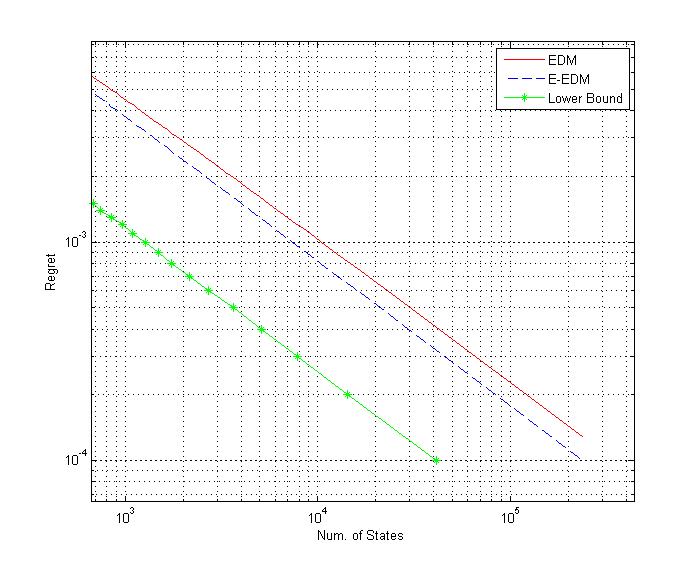}
 \caption[Regret vs. number of states - E-EDM, EDM machines and the lower bound] {Comparing the performance of the E-EDM machine, the EDM machine and the lower bound.
 \label{fig:RvsNumStatesAsyPerf}}
\end{figure}

\section{Summary and conclusions}
\label{sec:Summary}

In this paper we studied the problem of predicting an individual continuous sequence as well as the empirical mean with finite-state machine.

For small number of states, or equivalently, when the desired maximal regret is relatively large, we presented a new class of machines, termed the Degenerated Tracking
Memory (DTM) machines. An algorithm for constructing the best predictor among this class was given. For small enough number of states, this optimal DTM machine was shown to be optimal among \textbf{all} machines. It is still unknown up to which number of states this result holds true. Nevertheless, for larger number of states, one can try to attain better performance by easing the constraints imposed on the class of DTM machines and allowing more than a single state down-jump (up-jump) from all states in the lower (upper) half. The construction of the optimal machine in that case is, however, much more complex. Another important implication of these restrictions, is a lower bounded of $R=0.0278$ on the achievable maximal regret of any DTM machine.

For universal predictors with a large number of states, or equivalently, when the desired maximal regret is relatively small, we proved a lower bound of $O(k^{-2/3})$ on the maximal regret of any $k$-states machine. We proposed the Exponential Decaying Memory (EDM) machine and showed that the worst sequence incurs a bounded regret of $O(k^{-2/3})$, where $k$ is the number of states. We further presented the Enhanced Exponential Decaying Memory (E-EDM) machine which outperforms the EDM machine and better approaches to the lower bound. An interesting observation is that both machines are equivalent up to the prediction values, where a better state allocation is preformed when constructing the E-EDM machine. Recalling that the EDM machine is a finite-memory approximation of the Cumulative Moving Average predictor which is the best unlimited resources universal predictor (w.r.t. the non-universal empirical mean predictor) \cite{UniversalSchemes}, we can understand why both the EDM and the E-EDM machines approach optimal performance.

Analyzing the performance of the EDM and the E-EDM machines showed that the regret of any sequence can be upper bounded by the sum of two losses - {\em quantization loss}, the loss incurred by the quantization of the input samples, and {\em spacing loss}, the loss incurred by the quantization of the prediction values. It is worth mentioning that the maximal regret of the optimal DTM machine can also be upper bounded by the sum of these losses. As the number of states in the optimal DTM machine increases, the {\em quantization loss} goes to the lower bound, $R=0.0278$, and the {\em spacing loss} goes to zero. Thus, understanding the optimal allocation between these two losses may lead to the answer of up to which number of states the optimal DTM machine is the best universal predictor. It is also worth mentioning that the E-EDM machine is constructed with allocating half of the desired regret to the {\em quantization loss} and the other half to the {\em spacing loss}. A further optimization may be obtained by a different allocation.

Throughout this paper we assumed that the sequence's outcomes are bounded. Note that this constraint is mandatory since the performance of a universal predictor is analyzed by the regret of the worst sequence. In the unbounded case, for any finite-memory predictor one can find a sequence that incurs an infinite regret. However, an optional further study is to expand the results presented here to a more relaxed case, e.g. sequences with a bounded difference between consecutive outcomes.

In this study we essentially examined finite-memory universal predictors trying to attain the performance of the (non-universal) ``zero-order'' predictor, i.e., the empirical variance of any individual continuous sequence. We believe that our work is the first step in the search for the best finite-memory universal predictor trying to attain the performance of the best (non-universal) $L$-order predictor, for any $L$.

\appendices
\section{Proof of the lower bound given in Theorem \ref{thm:EDMregret}}
\label{app:lowerBoundTheorem}

\begin{proof}
Here we show that there is a continuous-valued sequence which rotates the EDM
machine (denoted $U_{EDM}$) in a minimal circle incurring a regret of
$\frac{1}{2}k^{-2/3}+O(k^{-1})$.

Consider the following minimal circle - $m$ states up-step, $m-1$
states down-step, $m$ states up-step, $m-1$ states down-step and
so on $m-1$ times. The last step is a down-step of $m-1$ states that
close the circle and return the machine to the initial state.
Denoting the states' gap by $\Delta$, the described sequence can be written as follows\footnote{Note that we can always apply $\xi>0$ as small as desired to ensure that the samples are not exactly equal to the transition threshold, but otherwise inside the regions of transition. For example, we could have taken $x_1=\hat{x}_1+(m+\tfrac{1}{2}-\xi)\Delta k^{2/3}$ with $\xi\rightarrow 0$.}:
\begin{align*}
& x_1=\hat{x}_1+(m+\tfrac{1}{2})\Delta k^{2/3} \\
& x_2=\hat{x}_1+m\Delta-(m-1-\tfrac{1}{2})\Delta k^{2/3} \\
& x_3=\hat{x}_1+\Delta+(m+\tfrac{1}{2})\Delta k^{2/3} \\
&\vdots\\
& x_{2m-3}=\hat{x}_1+(m-2)\Delta+(m+\tfrac{1}{2})\Delta k^{2/3} \\
& x_{2m-2}=\hat{x}_1+(2m-2)\Delta-(m-1-\tfrac{1}{2})\Delta k^{2/3} \\
& x_{2m-1}=\hat{x}_1+(m-1)\Delta-(m-1-\tfrac{1}{2})\Delta k^{2/3}~.
\end{align*}
Now, assuming that all of these sample are between $0$ and $1$, one can note that they form a minimal circle of $2m-2$ states $\{\hat x_1,\ldots,\hat x_{2m-1}\}$ with equal $\Delta$ spacing between them. The circle is as follows: $\hat x_1 \hookrightarrow \hat x_{m+1} \mapsto \hat x_2 \hookrightarrow \hat x_{m+2} \mapsto \hat x_{3} \hookrightarrow \ldots \hookrightarrow \hat x_{2m-1} \mapsto \hat x_{m} \mapsto \hat x_1$, where $\hookrightarrow$ and $\mapsto$ denote up and down-step, accordingly.

Analyzing the regret of the described sequence results in
\begin{align}
R(U_{EDM},x_1^{2m-1})&=\Delta^2(\tfrac{1}{4} k^{4/3}+m(m-1)k^{2/3}-\tfrac{m(m-1)}{3}). \label{appA:eq1}
\end{align}

Let us choose
\begin{align}
m=\lfloor\tfrac{\tfrac{1}{2}k^{-2/3}}{\Delta}\rfloor  \label{appA:eq2}~,
\end{align}
where $\lfloor x \rfloor$ denotes the rounding of $x$ to the  largest previous integer. In that case the highest sample, $x_{2m-3}$, satisfies $x_{2m-3}\leq \hat{x}_1+1/2k^{-2/3}-2\Delta+1/2+1/2k^{-1/3}$, and the lowest sample $x_{2m-1}$, satisfies $x_{2m-1}\geq \hat{x}_1+1/2k^{-2/3}-2\Delta-1/2+3/2k^{-1/3}$.  Choosing, for example,
\[~\hat{x}_1=Q(\tfrac{1}{2}-\tfrac{1}{2}k^{-1/3}-\tfrac{1}{2}k^{-2/3}+\Delta)~,\] where $Q(\cdot)$ denotes the quantization to the nearest state, results $x_{2m-3}\leq 1$ and $x_{2m-1}\geq 0$, and thus all samples $\{ x_1,\ldots, x_{2m-1}\}$ are valid, that is, satisfy $0\leq x_t \leq 1$.

Now, by applying Equation \eqref{appA:eq2} into Equation
\eqref{appA:eq1} we get
\begin{align}
R(U_{EDM},x_1^{2m-1})&=\tfrac{1}{4}\Delta^2 k^{4/3}+\tfrac{1}{4} k^{-2/3}+O(k^{-1}) \nonumber\\
&= \tfrac{1}{2}k^{-2/3}+O(k^{-1})~.
\end{align}
\end{proof}

\section{Proof of Lemma \ref{lemMinStep}}
\label{app:LemmalemMinStep}
\begin{proof}
Consider a sequence $x_1,...,x_{L+1}$ that rotates an FSM, denoted $U$, in a minimal circle, where $x_1$ induces a single up-jump of $L$ states and $x_2^{L+1}$ induce down-jumps of a single state. Since the regret of any zero-step minimal circle is smaller than $R$, an input sample that satisfies $x=\hat{x}_t-\sqrt{R}-\varepsilon$, where $\varepsilon\rightarrow0^+$, must induce a down-jump of at least one state. Thus, we can always choose the input samples $x_2^{L+1}$ to satisfies $x_t \geq \hat{x}_t-\sqrt{R}$. We shall also assume that $x_1$ satisfies:
\begin{equation}
x_1 > \hat{x}_1+(1+2L)\sqrt{R}~,
\end{equation}
where $\hat{x}_1=S_i$. \emph{\textbf{We show that this assumption can not hold true}}.

By denoting $\lambda_t=\hat{x}_t-\hat{x}_1$ we note that the empirical mean of the sequence satisfies:
\begin{align}
\bar{x} & \geq \hat{x}_1+\sqrt{R}+\tfrac{1}{L+1}\sum_{t=1}^{L+1}\lambda_t ~.\label{lemMinStep:eq3}
\end{align}
Now, let us examine the regret incurred by the described sequence:
\begin{align}
R(U,x_1^L) & = \tfrac{1}{L+1}\sum_{t=1}^{L+1}(x_t-\hat{x}_t)^2-(x_t-\bar{x})^2 \nonumber\\
&= (\bar{x}-\hat{x}_1)^2+\tfrac{1}{L+1}\sum_{t=1}^{L+1}\lambda_t^2-2\lambda_t(x_t-\hat{x}_1) \nonumber\\
& {\geq}(\bar{x}-\hat{x}_1)^2-\tfrac{1}{L+1}\sum_{t=1}^{L+1}\lambda_t^2 \label{lem6:eq1}\\
& {\geq} (\sqrt{R}+\tfrac{1}{L+1}\sum_{t=1}^{L+1}\lambda_t)^2-\tfrac{1}{L+1}\sum_{t=1}^{L+1}\lambda_t^2 \label{lem6:eq2}\\
& > R+\tfrac{1}{L+1}\sum_{t=1}^{L+1}(2\sqrt{R}-\lambda_t)\lambda_t~,\label{lemMinStep:eq5}
\end{align}
where \eqref{lem6:eq1} follows $~\lambda_t\geq 0$ and $x_t\leq \hat{x}_t$ for all the down samples $x_2^{L+1}$, \eqref{lem6:eq2} follows \eqref{lemMinStep:eq3}. In \cite{IngberThesis} it is shown that in an FSM with maximal regret $R$ w.r.t. binary sequences, the maximal up-jump is no more than $2\sqrt{R}$. Therefore, this must hold also for continuous-valued sequences. Hence, in the discussed minimal circle all states are within $2\sqrt{R}$ from the initial state, that is $2\sqrt{R}\geq \lambda_t$ for all $t$ and we get $R(U,x_1^L)>R$.

We can now conclude that to attain a regret smaller than $R$, any input sample $x$ that induces an $L$ states up-jump from state $i$, must satisfy:
\begin{equation}
x \leq S_i+(1+2L)\sqrt{R}~.
\end{equation}
Since an input sample $1$ induces an $m_{u,i}$ states jump from state $i$ we conclude that the following must be satisfied:
\begin{equation}
 1 \leq S_i+(1+2m_{u,i})\sqrt{R}~.
 \end{equation}
In the same manner it can be shown that $0 \geq
S_i-(1+2m_{d,i})\sqrt{R}$.\\
\end{proof}

\section{Proof of Lemma \ref{lem:spacingLossWithinSegment}}
\label{app:LemmaspacingLossWithinSegment}

\begin{proof}
First we note that:
\begin{align}
&-\tfrac{1}{L}\sum_{t=1}^L P_t(\hat{x}_t-\hat{x}_1)=-\tfrac{1}{L}\sum_{t=1}^L P_t\hat{x}_t~,
\end{align}
where we used $\sum_{t=1}^L P_t=0$. Note that $P_t\hat{x}_t$ is positive for up-steps and
negative for down-steps. We consider a minimal circle within a segment $A(m_u,m_d)$ that crosses states with the same spacing gap, denoted $\Delta=\Delta(m_u,m_d)$. It follows that:
\begin{align}
&-\tfrac{1}{L}\sum_{t=1}^L P_t(\hat{x}_t-\hat{x}_1) = -\tfrac{1}{L}\sum_{t=1}^L P_t\sum_{j=1}^{t-1}P_j \Delta ~.\nonumber
\end{align}

Define {\em{\textbf{mixed}}} sequences as sequences where the up and down
steps are interlaced. Define {\em{\textbf{straight}}} sequences as sequences
where all the up-steps are first, followed by all the down-steps (consecutive in time). We show that any {\em{\textbf{mixed}}} sequence with $\{P_t\}_{t=1}^L$ jumps that rotates the machine in a minimal circle with the same spacing gap for all states can be transformed into a {\em{\textbf{straight}}} sequence with the same jumps only in a different order (up-jumps are first) without changing the {\em spacing loss} of the sequence. First we note that for any three interlaced jumps
\[\text{up jump $\rightarrow$ down jump $\rightarrow$ up jump},\]
that cross
\[P_{u,1}~ \rightarrow~ P_d ~\rightarrow~ P_{u,2}~\]
states (accordingly), the following holds true:
\begin{align}
&P_{u,1}\hat{x}_{u,1}+P_d(\hat{x}_{u,1}+P_{u,1}\Delta)+ \nonumber\\
& \qquad \qquad +P_{u,2}(\hat{x}_{u,1}+(P_{u,1}+P_d)\Delta) \nonumber\\
&\qquad = P_{u,1}\hat{x}_{u,1}+P_{u,2}(\hat{x}_{u,1}+ \nonumber\\
& \qquad \qquad+P_{u,1}\Delta)+P_d(\hat{x}_{u,1}+(P_{u,1}+P_{u,2})\Delta)~. \label{straightSeq}
\end{align}
Thus, Equation \eqref{straightSeq} implies that the {\em spacing loss} of these three jump does not change when the order of the jumps is:
\[\text{up jump $\rightarrow$ up jump $\rightarrow$ down jump}.\]
This can be shown also for a sequence with more than one consecutive down-jumps between two up-steps:
\[\text{up jump $\rightarrow$ down jump $\rightarrow$ ... $\rightarrow$ down jump $\rightarrow$ up jump}~.\]
Hence, in a recursive way any {\em{\textbf{mixed}}} sequence can be transformed into a {\em{\textbf{straight}}} sequence without changing the {\em spacing loss} by moving all the down-jumps to the end of the sequence. In the rest of the proof we shall assume {\em{\textbf{straight}}} sequences. Note that this transformation changes the states of the minimal circle, but since we transform the sequence only for an easier analyze, we can assume that all states still have the same spacing gap. Figure \ref{fig:EEDM2_mixedStraight} gives an example.

\begin{figure}[ht]
\centering
\includegraphics[width=0.8\columnwidth,height=0.15\textheight]{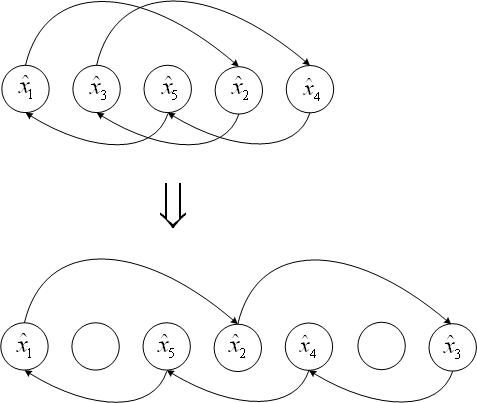}
\caption[{\em{Mixed}} - {\em{Straight}} sequences]{An example for a {\em{\textbf{mixed}}} sequence transformed into a {\em{\textbf{straight}}} sequence. \label{fig:EEDM2_mixedStraight}}
\end{figure}

We continue by proving that applying maximum up and down steps
maximize the {\em spacing loss}. Consider two consecutive down-steps
of $P_{d_1},P_{d_2}$ states staring at state $\hat{x}$, with a total of $C$ states, i.e $\abs{P_{d_1}}+\abs{P_{d_2}}=C$. Note that we examine two down-steps, thus $C\leq 2m_d$. The
{\em spacing loss} of these two down-steps is:
\begin{equation}
\hat{x}\cdot \abs{P_{d,1}}+(\hat{x}-\abs{P_{d,1}}\Delta)\cdot \abs{P_{d,2}}=\hat{x}\cdot
C-\abs{P_{d,1}}(C-\abs{P_{d,1}})\Delta~.
\end{equation}
If $C\leq m_d$ the {\em spacing loss} is maximized for $\abs{P_{d,1}}=C$
and $\abs{P_{d,2}}=0$. If $m_d \leq C \leq 2m_d$ then the
{\em spacing loss} is maximized for
$\abs{P_{d,1}}=m_d$. We got that we can maximize the {\em spacing loss} by taking a couple of
down-steps and unite them into a single down-step (if together they cross no more than $m_d$ states), or to apply maximum down-step, $m_d$, to the first and $C-m_d$ to the second (if together they cross more than $m_d$ states). Thus, assuming {\em{\textbf{straight}}} sequences, we can start with the first couple of down-steps, maximize the {\em spacing loss} by applying maximum down-step, then take the third down-step and apply maximum down-step with the new down-steps that were created. In a recursive way we can maximize the {\em spacing loss} by applying maximum down-steps (note that the number of down-steps
reduces which also maximize the {\em spacing loss}). In the same manner it can be shown that applying maximum up-steps maximize the {\em spacing loss}.

\begin{figure}[ht]
\centering
\includegraphics[width=0.4\columnwidth,height=0.06\textheight]{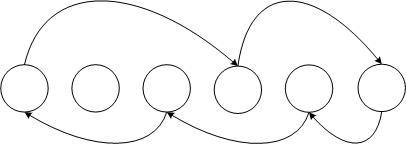}
\caption[Worst case spacing loss - an example]{An example for the worst case {\em spacing loss} of a minimal circle that crosses $5$ states in the segment $A(3,2)$. \label{fig:EEDM3_worstCaseSpacingloss}}
\end{figure}

Consider a minimal circle of $C$ states crossed on the way up and down, all in the segment $A(m_u,m_d)$. The worst case scenario for the {\em spacing loss} is composed of $N_u$ up-steps each of $m_u$ states jump (maximum up-jump), a single up-step of $c_u$ states, where $c_u=mod(C,m_u)$, $N_d$ down-steps each of $m_d$ states jump (maximum down-jump), and a single down-step of $c_d$ states, where $c_d=mod(C,m_d)$. $N_d$ and $N_u$ satisfy $C=N_um_u+c_u$ and $C=N_dm_d+c_d$. It can be shown that the position in the sequence of the single up-step (of $c_u$ states) and the single down-step (of $c_d$ states) has no impact on the {\em spacing loss}. Let us analyze the {\em spacing loss} of the \em{straight} sequence. First, all up-steps satisfy:
\begin{align}
-\tfrac{1}{L}\sum_{t\in \text{\{up steps\}}} &P_t(\hat{x}_t-\hat{x}_1)= \nonumber\\
&=-\tfrac{1}{L}\Delta(\sum_{i=0}^{N_u-1}m_u(i\cdot m_u)+N_um_uc_u)\nonumber\\
&=-\tfrac{1}{L}\Delta(m_u^2\tfrac{N_u(N_u-1)}{2}+N_um_uc_u)\nonumber\\
&=-\tfrac{1}{L}\tfrac{\Delta}{2}(C^2-m_uC+c_u(m_u-c_u))~.
\end{align}
In the same manner, all down-steps satisfy:
\begin{align}
-\tfrac{1}{L}\sum_{t\in \text{\{down steps\}}} &P_t(\hat{x}_t-\hat{x}_1)= \nonumber\\
&=\tfrac{1}{L}\Delta(\sum_{i=1}^{N_d}m_d(i\cdot m_d)+c_dC)\nonumber\\
&=\tfrac{1}{L}\tfrac{\Delta}{2}(C^2+m_dC-c_d(m_d-c_d))~.
\end{align}
Thus, the worst case scenario of the {\em spacing loss} satisfies:
\begin{align}
-\tfrac{1}{L}\sum_{t=1}^L &P_t(\hat{x}_t-\hat{x}_1)= \nonumber\\
&=\tfrac{1}{L}\tfrac{\Delta}{2}(C(m_u+m_d)-c_u(m_u-c_u)-c_d(m_d-c_d)) \label{eq:forLemma7}\\
&\leq \tfrac{1}{L}\tfrac{\Delta}{2}C(m_u+m_d)~,
\end{align}
where the length of the circle satisfies:
\begin{equation}
L=\lceil\tfrac{C}{m_u}\rceil+\lceil\tfrac{C}{m_d}\rceil \geq \tfrac{C}{m_u}+\tfrac{C}{m_d}~.
\end{equation}
Therefore, the worst case scenario satisfies:
\begin{align}
-\tfrac{1}{L}\sum_{t=1}^L P_t(\hat{x}_t-\hat{x}_1)&\leq \tfrac{m_um_d}{2}\Delta~.
\end{align}
Since $\Delta= \Delta(m_u,m_d)=\tfrac{\sqrt{R}}{2m_um_d}$ we get that the {\em spacing loss} for any minimal circle within a segment (and with identical spacing gap between all states) satisfies:
\begin{align}
\text{\em spacing loss} \leq 4\sqrt{R}\tfrac{m_um_d}{2}\Delta(m_u,m_d) = R~.
\end{align}
\end{proof}

\section{Proof of Lemma \ref{lem:spacinglossSegments}}
\label{app:LemmaspacinglossSegments}

\begin{proof}
We denote two adjacent segments by $A(m_{u,1},m_{d,1})$ and $A(m_{u,2},m_{d,2})$. Assume $A(m_{u,1},m_{d,1})$ is the lower segment and the minimal circle starts at the lowest state. Denote the spacing gap of each segment by $\Delta_1=\Delta(m_{u,1},m_{d,1})$ and $\Delta_2=\Delta(m_{u,2},m_{d,2})$. Note that if $\Delta_1<\Delta_2$ then $m_{u,2}=m_{u,1}-1~,~m_{d,2}=m_{d,1}$ and if $\Delta_1 > \Delta_2$ then $m_{u,2}=m_{u,1}~,~m_{d,2}-1=m_{d,1}$.

\begin{figure}[ht]
\centering
\includegraphics[width=1\columnwidth]{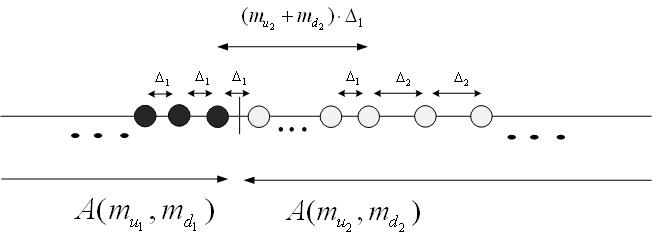}
\caption[Spacing gap in the E-EDM machine]{Spacing gap between states in the connection between the segments $A(m_{u,1},m_{d,1})$ and $A(m_{u,2},m_{d,2})$. See the E-EDM machine definitions in section \ref{sec:DesigningEEDM}. \label{fig:appA_spacingGap}}
\end{figure}

First we assume that the minimal circle traverse between the segments only once (that is, once on the way up and once on the way down). We also assume that $\Delta_1<\Delta_2$. We can now divide the minimal circle into two virtual minimal circles - take the up-step that traverse the machine to the higher segment and denote the destination state of this jump by $\hat{x}_c$. Take a down-step that crosses state $\hat{x}_c$ and split it into two steps - assuming the down-step crosses $P_d$ states, $c_d$ states jump to state $\hat{x}_c$ and $(P_d-c_d)$ states jump from state $\hat{x}_c$. Note that two minimal circles were constructed - left minimal circle that traverse $C_1$ states and right minimal circle that traverse $C_2$ states. This is depict in Figure \ref{fig:appA_splitDownStep}. The {\em spacing loss} of the down-step satisfies:
\begin{equation}
P_d(\hat{x}_c+c_d\Delta_1)=c_d(\hat{x}_c+c_d\Delta_1)+(P_d-c_d)\hat{x}_c+(P_d-c_d)c_d\Delta_1~. \label{eq:downStep}
\end{equation}

\begin{figure}[ht]
\centering
\includegraphics[width=1\columnwidth]{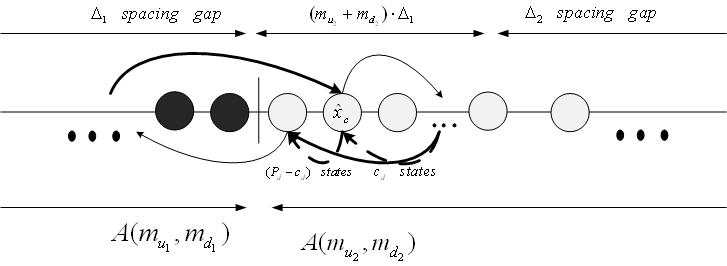}
\caption[Minimal circle between segments - splitting a down-step]{Minimal circle that traverse once between segments. Splitting the marked down-step that crosses  state $\hat{x}_c$ into two down-steps, creating two virtual minimal circles to the right and left. Note that since the first $m_{u,2}+m_{d,2}$ states at the second segment are with spacing gap $\Delta_1$, the marked down-step must only cross states with spacing gap $\Delta_1$. \label{fig:appA_splitDownStep}}
\end{figure}

Note that $\hat{x}_c$ is in the upper segment but we used $\Delta_1$ since the first $m_{u,2}+m_{d,2}$ states in the upper segment have spacing gap of $\Delta_1$ (see the construction of the E-EDM machine in section \ref{sec:DesigningEEDMAlgo}). Also note that the first term in the right hand side of Equation \eqref{eq:downStep} belongs to the {\em spacing loss} of the right minimal circle and the middle term belongs to the {\em spacing loss} of the left minimal circle. Note that the {\em spacing loss} of the minimal circle is compose of the {\em spacing loss} of the left and right minimal circles and the last term in Equation \eqref{eq:downStep}. The left minimal circle traverse $C_1$ states, all with spacing gap $\Delta_1$. The right minimal circle traverse $C_2$ states, some with spacing gap $\Delta_1$ and some with $\Delta_2$. We can now conclude that the {\em spacing loss} satisfies:
\begin{align}
\text{\em spacing} & \text{\em ~loss}\leq 4\sqrt{R}\tfrac{1}{L} \big(~[C_1(m_{u,1}+m_{d,1}) \nonumber \\
&-(P_d-c_d)(m_{d,1}-(P_d-c_d))]\tfrac{\Delta_1}{2}\nonumber\\
&+[C_2(m_{u,2}+m_{d,2})-c_d(m_{d,2}-c_d)]\tfrac{\Delta_2}{2} \nonumber\\
&+c_d(P_d-c_d)\Delta_1~\big)~, \label{eq:app1}
\end{align}
where we applied Lemma \ref{lem:spacingLossWithinSegment} (Equation \eqref{eq:forLemma7}) to bound the {\em spacing loss} of the left and right minimal circles. Note that Lemma \ref{lem:spacingLossWithinSegment} is true for the right minimal circle since all states have a spacing gap that is no more than $\Delta_2$. Now, since $m_{d,1}=m_{d,2}$ and $\Delta_1<\Delta_2$ we get:
\begin{align}
\text{\em spacing loss}&\leq 4\sqrt{R}\tfrac{1}{L}(C_1(m_{u,1}+m_{d,1})\tfrac{\Delta_1}{2}+ \nonumber\\
&\qquad \qquad +C_2(m_{u,2}+m_{d,2})\tfrac{\Delta_2}{2}) \nonumber\\
&=R\tfrac{1}{L}(\tfrac{C_1}{m_{d,1}}+\tfrac{C_1}{m_{u,1}}+\tfrac{C_2}{m_{d,2}}+\tfrac{C_2}{m_{u,2}})~. \label{eq:app2}
\end{align}
Let us bound the length of the minimal circle:
\begin{align}
L &\geq \lceil\tfrac{C_1}{m_{u,1}}\rceil +\lceil\tfrac{C_2}{m_{u,2}}\rceil+\lceil\tfrac{C_1+C_2}{m_{d,1}}\rceil \nonumber\\
&\geq \tfrac{C_1}{m_{u,1}} +\tfrac{C_2}{m_{u,2}}+\tfrac{C_1+C_2}{m_{d,1}}~.
\end{align}
Applying this into Equation \eqref{eq:app2} results:
\begin{equation}
\text{\em spacing loss}\leq R~.
\end{equation}

Assume again that the minimal circle traverse between the segments only once but now assume $\Delta_1>\Delta_2$. Divide the minimal circle into two virtual minimal circles in the same manner as above but now take the down-step that traverse the machine to the lower segment and split an up-step. In the same manner we can show that the {\em spacing loss} is not more than R.

If assuming that the minimal circle traverse between segments $m$ times, in the same manner as above we can divide the circle into $m$ left minimal circles and $m$ right minimal circles and bound the {\em spacing loss}.\\
\end{proof}

\section{Proof of Theorem \ref{thm:nStatesEEDM}}
\label{app:TheoremnStatesEEDM}

\begin{proof}
Consider an E-EDM machine that was designed to attain maximal regret $R_d$. By denoting $R=\frac{R_d}{2}$, the number of states satisfies:
\begin{equation} \label{eq:nEEDMStates1}
k \leq \sum_{m_u,m_d \in \mathbb{N}}(\tfrac{\abs{A(m_u,m_d)}}{\Delta(m_u,m_d)}+2)~,
\end{equation}
where all states in the segment $A(m_u,m_d)$ have a maximum up and down step of $m_u,~m_d$ states and $\Delta(m_u,m_d)$ spacing gap. As shown in the definitions of the E-EDM machine in section \ref{sec:DesigningEEDM}, we add to each segment at most two states to ensure regret smaller than $R_d$ for sequences that rotate the E-EDM machine in a minimal circle that traverse between segments. Note that there are at most $\lceil\tfrac{1}{2\sqrt{R}}\rceil$ segments.\\
Let us examine Equation \eqref{eq:nEEDMStates1}:
\begin{align}
k &\leq R^{-1/2}+2+\sum_{m_u,m_d \in \mathbb{N}}\tfrac{\abs{A(m_u,m_d)}}{\Delta(m_u,m_d)} \nonumber\\
& =R^{-1/2}+2+\sum_{m_u,m_d \in \mathbb{N}}\tfrac{\abs{A(m_u,m_d)}}{\sqrt{R}}2 m_u m_d \nonumber\\
& =R^{-1/2}+2+2 R^{-1/2}\sum_{m_u,m_d \in \mathbb{N}}\abs{A(m_u,m_d)}\cdot \nonumber\\
& \qquad \qquad \qquad \qquad \qquad \cdot \lceil\tfrac{1-x-\sqrt{R}}
{2\sqrt{R}}\rceil \cdot \lceil\tfrac{x-\sqrt{R}}{2\sqrt{R}}\rceil \Big|_{x\in A(m_u,m_d)} \nonumber\\
& \leq R^{-1/2}+2+\tfrac{1}{2} R^{-3/2}\sum_{m_u,m_d \in \mathbb{N}}\abs{A(m_u,m_d)}\cdot \nonumber\\
& \qquad \qquad \qquad \qquad \qquad \cdot \big(x(1-x)+\sqrt{R}+R\big)\Big|_{x\in A(m_u,m_d)}~.
\end{align}
By denoting the segments with the same maximum up-step as
$B(m_u)$, we can further bound the number of states:
\begin{align}
k & \leq \tfrac{1}{2} (R^{-1} +3R^{-1/2})+2+\tfrac{1}{2}R^{-3/2}\sum_{m_u \in \mathbb{N}}\abs{B(m_u)}\cdot \nonumber\\
& \qquad \qquad \qquad \qquad \qquad \cdot \max_{x\in B(m_u)}x(1-x)~.
\end{align}
Since $\abs{B(m_u)}=2\sqrt{R}$ for almost all $m_u$
($\abs{B(m_u)}\leq 2\sqrt{R}$ at the edges of the interval
$[0,\tfrac{1}{2}]$), $x(1-x)$ is a concave function with a singular maximum point at $\frac{1}{2}$ and the number of states in the lower and upper halves is equal, we get:
\begin{align}
k &\leq \tfrac{1}{2} \big(R^{-1} +3R^{-1/2}\big)+2+ \nonumber \\
&\qquad +R^{-3/2}\sum_{i=1}^{\lceil\tfrac{1}{4\sqrt{R}}\rceil}2\sqrt{R}(\sqrt{R}+i2\sqrt{R})(1-(\sqrt{R}+i2\sqrt{R})) \nonumber\\
&\leq \tfrac{1}{12}R^{-3/2}-\tfrac{5}{12}R^{-1}-12R^{-1/2}-32 \nonumber\\
&= \tfrac{2^{3/2}}{12}R_d^{-3/2}+O(R_d^{-1})~,
\end{align}
where we applied $R=\frac{R_d}{2}$.

We can also bound the number of states from below by:
\begin{align}
k &\geq \sum_{m_u,m_d \in \mathbb{N}}\tfrac{\abs{A(m_u,m_d)}}{\Delta(m_u,m_d)} \nonumber\\
& \geq \tfrac{1}{2} R^{-3/2}\sum_{m_u,m_d \in \mathbb{N}}\abs{A(m_u,m_d)}\cdot \big(x(1-x)- \nonumber\\
&\qquad \qquad \qquad \qquad \qquad -\sqrt{R}+R\big)\Big|_{x\in A(m_u,m_d)}~.
\end{align}
By denoting the segments with the same maximum up-step as
$B(m_u)$, we can bound the number of states from below:
\begin{align}
k & \geq \tfrac{1}{2} \big(-R^{-1} +R^{-1/2}+ \nonumber\\
& \qquad \qquad \qquad +R^{-3/2}\sum_{m_u \in \mathbb{N}}\abs{B(m_u)}\cdot \min_{x\in B(m_u)}x(1-x)\big)~.
\end{align}
Using the approximation we made to calculate the lower bound we
get:
\begin{align}
k &\geq \tfrac{1}{12} (R^{-3/2}-15R^{-1}+2R^{-1/2}) \nonumber\\
&= \tfrac{1}{12}(\tfrac{R_d}{2})^{-3/2}+O(R_d^{-1})~.
\end{align}
Thus, we upper and lower bounded the number of states in the E-EDM machine by $\tfrac{1}{12}(\tfrac{R_d}{2})^{-3/2}+O(R_d^{-1})$.
\end{proof}


\bibliographystyle{IEEEtran}
\bibliography{mybib2}

\end{document}